\documentclass[11pt,review]{elsarticle}
\journal{....}
\usepackage{lineno,hyperref}
\modulolinenumbers[5]
\usepackage{amssymb}
\usepackage{moreverb,multirow}
\usepackage{amssymb,amsmath,amsthm}
\usepackage{amsfonts,dsfont,mathtools}
\usepackage[margin=2cm]{geometry}
\usepackage[hypcap]{caption}
\usepackage{caption}
\usepackage{subcaption}
\usepackage{multicol}
\usepackage{morefloats}
\usepackage{booktabs,tabularx}
\usepackage{multirow}
\usepackage{url}
\usepackage{graphics, graphicx, epstopdf}
\usepackage{algorithm}
\theoremstyle{plain}

\newtheorem{thm}{Theorem}[section]
\theoremstyle{definition}

\newtheorem{rem}{Remark}[section]
\newtheorem{cor}{Corollary}[section]
\newtheorem{prop}{Proposition}[section]

\theoremstyle {plain}
\newtheorem{lem}[thm]{Lemma}

\usepackage{color,xcolor}
\begin{document}
	
\begin{frontmatter}
\title{Impact of imperfect vaccine, vaccine trade-off and population turnover on infectious disease dynamics}
		
		
	\author[mymainaddress1,mymainaddress2]{Hetsron L. Nyandjo-Bamen}
			\author[mymainaddress1]{Jean Marie Ntaganda}
		\author[mymainaddress3]{Aurélien Tellier}
		\cortext[mycorrespondingauthor]{Corresponding authors}
		\author[mymainaddress4,mymainaddress2]{Olivier Menoukeu-Pamen}
		\address[mymainaddress1]{Department of Mathematics, School of Science, College of Science and Technology, University of Rwanda, Rwanda}
		\address[mymainaddress2]{African Institute for Mathematical Sciences, Ghana}
		\address[mymainaddress3]{Population Genetics, Department of Life Science systems, School of Life Sciences, Technical University of Munich  85354 Freising, Germany}
		\address[mymainaddress4]{IFAM, Department of Mathematical Sciences, University of Liverpool, United Kingdom}
		
		\vspace{-0.5cm}
\begin{abstract}
Vaccination is essential for the management of infectious diseases, many of which continue to pose devastating public health and economic challenges across the world. However, many vaccines are imperfect having only a partial protective effect in decreasing disease transmission and/or favouring recovery of infected individuals, and possibly exhibiting trade-off between these two properties. Furthermore, population turnover, that is the rate at which individuals enter and exit the population, is another key factor determining the epidemiological dynamics. While these factors have yet been studied separately, we investigate the interplay between the efficiency and property of an imperfect vaccine and population turnover. We build a mathematical model with frequency incidence rate, a recovered compartment, and an heterogeneous host population with respect to vaccination. We first compute the basic reproduction number $\mathcal{R}_0$ and study the global stability of the equilibrium points. Using a sensitivity analysis, we then assess the most influential parameters determining the total number of infected and $\mathcal{R}_0$ over time. We derive analytically and numerically conditions for the vaccination coverage and efficiency to achieve disease eradication ($\mathcal{R}_0 < 1$) assuming different intensity of the population turnover (weak and strong), vaccine properties (transmission and/or recovery) and trade-off between the latter. We show that the minimum vaccination coverage increases with lower population turnover, decreases with higher vaccine efficiency (transmission or recovery), and is increased/decreased by up to 15\% depending on the trade-off between the vaccine properties. We conclude that the coverage target for vaccination campaigns should be evaluated based on the interplay between these factors.  

\end{abstract}
		
		
	\end{frontmatter}

\newpage
\section{Introduction}
\label{sec:introduction}
{\allowdisplaybreaks
Vaccination is one of the most effective public health policies for protecting humans and animals from infectious diseases. Global vaccination campaigns have helped eradicate diseases such as smallpox, measles, poliomyelitis, rinderpest in most parts of the world, ultimately saving the lives of millions of humans and animals. A perfect vaccine would keep vaccinated people from becoming infected when exposed to the pathogen. An imperfect vaccine, one that does not prevent vaccinated individuals from becoming infected upon pathogen exposure, may still be beneficial in various ways \cite{Anderson1991}. For example, imperfect vaccines may provide benefits such as preventing infection, limiting parasite within-host growth and thus reducing the damage done to the host \cite{Vale2014}, or preventing transmission by infected hosts \cite{gandon2003imperfect}. As we have seen recently with the epidemic of Covid-19, imperfect vaccines can be used to reduce the number of infected individuals, but also to protect individuals at risk of developing the more lethal form of the infection, especially when the efficiency of vaccination may be volatile and decreases due to the appearance of new variants of the virus \cite{Hwang2021, Ioannidis2021, Dagan2021}.}

The effectiveness of a given vaccine is determined not only by its biochemical and immunological properties, but also by how the vaccine is deployed and what other health management (biosecurity) measures are in place. Maintaining herd immunity during a disease outbreak, for example, has been promoted as a highly effective disease control strategy \cite{Djatcha2017,Ashby2021, Mancusowill2021}. However, a continuous influx of new susceptible, possibly unvaccinated individuals contributes to the disease's long-term persistence in the population \cite{Scherer2002, Pulliam2007}. A frequent introduction of pathogen into a partially immune population with intermediate levels of population immunity can lead to an epidemic of longer duration and/or higher total number of infectious individuals than the introduction into a naive population \cite{Pulliam2007}.  This phenomenon is named as "epidemic enhancement" \cite{Pulliam2007}. More generally, the population turnover rate, that is the rate at which individuals can enter and exit the considered population, may affect the effectiveness of control strategies \cite{Knight2020}. In human but also domesticated animals, population turnover takes the form of immigration and emigration in and out of the population, as well as birth and death of individuals. The turnover is an often neglected factor in epidemiology when generalizing predictions of modelling from human to domesticated and animal populations. \\

Moreover, a second parameter of importance in studying the efficiency of vaccination strategies, is the existence of biological trade-offs in epidemiology. The prime example, is the trade-off between parasite virulence and transmission rate which raises challenges for vaccine manufacturing. Indeed, in the seminal paper by Gandon \textit{et al.} \cite{gandon2003imperfect}, it is predicted that vaccines affecting disease transmission may lead to a decrease of parasite virulence, while other types of vaccines (reducing within-host growth rate) may lead to an increase of parasite virulence, and thus the counter-effect of a worst epidemiological outcome. Interestingly, much work has been devoted to generate precise predictions for virulence evolution in known parasites by incorporating empirical characterizations of vaccine effects into models capturing the epidemiological details of a given system \cite{Gandon2008, Alizon2009, Cressler2016}. In contrast, biochemical and immunological trade-offs of the vaccine itself have received little attention. We mean here that vaccination can affect several aspects of the disease dynamics, such as within-host growth and transmission, with possible trade-offs between these characteristics. For example, a vaccine reducing within-host growth may be more or less effective in reducing disease transmission. We therefore generalize the definition of imperfect vaccines as providing partial protection (non-maximal efficiency) against infection (decreasing transmission), partially enhancing (not fully) recovery of infected individuals, and possible trade-off between these two properties. There has been remarkably little work done to generally assess how the interplay between different vaccine properties, trade-offs, and vaccination strategies influences the burden of the epidemic in an heterogeneous community with imperfect vaccination. \\

The aim of this study is therefore to assess, through mathematical modelling, whether the use of vaccines that decrease the infection is more efficient to eradicate the disease in an heterogeneous community than a vaccine that both reduces the infection and favours recovery, or a vaccine reducing the infection rate but favouring recovery. We also want to assess whether these results depend on the effect of population turnover, in order to generalize our results to animal populations. \\
The paper is organized as follows. First, the model is formulated in Section 2. We then compute the basic properties of the steady state solutions as well as the  existence of a local and global stability of the equilibrium points of the model (Section 3). We then provide a numerical sensitivity of the model and examples of numerical analyses for different parameter values describing the interaction between population turnover and vaccine trade-offs on the epidemiological outcome. We conclude by providing predictions on the applicability of these results to vaccination strategies in human but also domesticated animal species for which turnover rates represent different end of a continuum.

\section{Model formulation}
\label{subsec:modelformulation}
{\allowdisplaybreaks
The formulation of the model is based on compartmental modeling \cite{anderson2013compartment}, which consists in creating virtual reservoirs called compartments. A compartment is a kinetically homogeneous structure. This means that any individual who enters a compartment is identical, from the epidemiological point of view, to any other already present in that compartment. A mathematical model therefore consists of describing the flow of individuals between the various compartments.
	
To study the dynamic of an infectious disease during and after the vaccination campaign, we modify the model formulated in 	\cite{gandon2003imperfect} by adding a recovered compartment and we consider a frequency-dependent disease transmission (incidence rate). The model takes in to account only  host-to-host transmission of the disease. Since many vaccines do not guaranty a perfect immunity, we consider an heterogeneous host community with two types of hosts: fully susceptible to the disease, or partially resistant to infection due to the imperfect vaccination. The fully susceptible hosts consist of uninfected ($S_1$) and infected ($I_1$) individuals. And among the partially resistant hosts, we find the uninfected ($S_2$) and the infected ($I_2$) individuals. All infected individuals (fully susceptible or partially resistant) can become recovered ($R$), and all recovered individuals are fully immune to reinfection \cite{gandon2007evolutionary}. Thus, the total population at time $t$, $N(t)$ is given by
	$$N(t) = S_1(t) +S_2(t) + I_1(t)+ I_2(t) + R(t).$$
We assume the parasite population to be monomorphic (having only one type or genotype). We also assume that new uninfected hosts arise through birth and immigration at constant rate, $\theta$. Among these new uninfected, a proportion, $p$, is partially immune due to the vaccination, while the remaining proportion $1-p$ is completely vulnerable to the parasite. Uninfected, infected and recovered hosts die naturally at a rate, $\mu$ and infected hosts suffer additional mortality due to the virulence of the parasite. Since host resistance may reduce the impact of parasite \cite{gandon2003imperfect}, we assume the virulence of the parasite on  fully susceptible hosts, $d_1$, is greater than the one on partially resistant hosts, $d_2.$ Uninfected hosts become infected with the forces of infection $\lambda_1(t) = \beta_{11}\dfrac{I_1(t)}{N(t)}+\beta_{12}\dfrac{I_2(t)}{N(t)}$ \hspace{.2cm} \text{and}\hspace{.2cm} $\lambda_2(t) = \beta_{21}\dfrac{I_1(t)}{N(t)}+\beta_{22}\dfrac{I_2(t)}{N(t)}$ when they are fully susceptible or partially resistant, respectively. And since the resistance can decrease the probability of becoming infected \cite{gandon2003imperfect}, we assume $\beta_{21}\leqslant \beta_{11}$ and $\beta_{22}\leqslant \beta_{12}$. Recovery rates may differ between the fully susceptible, $\gamma_{1}$ and the partially resistant host, $\gamma_{2}$. The schematic diagram of the model is as shown  in Figure \ref{mod}.
	\begin{figure}[!h]
		\centering 
		\includegraphics[width=0.8\textwidth]{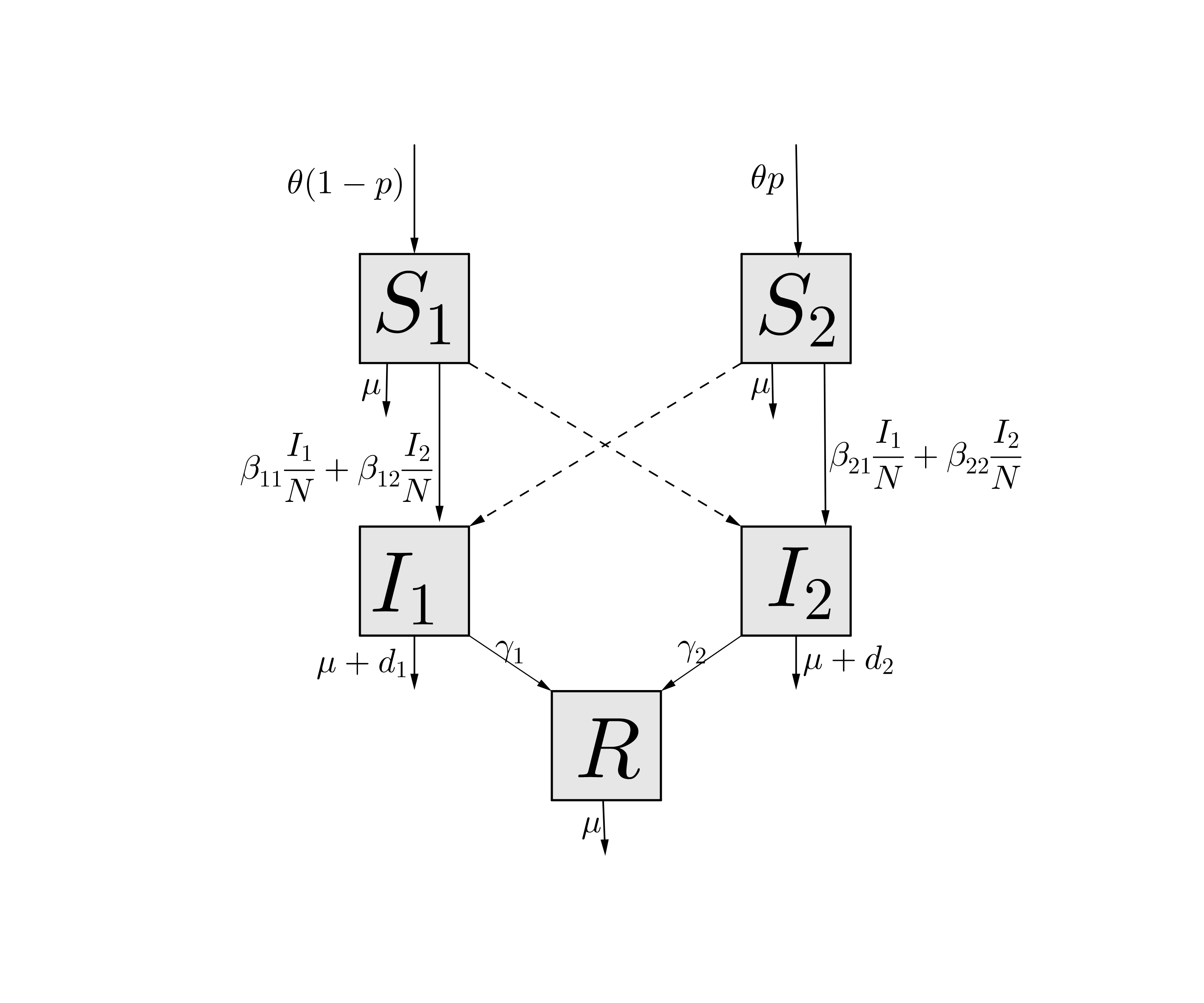}
		\caption{Schematic diagram of the epidemiological model with imperfect vaccination.}\label{mod}	
	\end{figure}
	Mathematically, the model is as follows: 
	\begin{equation}\label{deterministic}
		\left\{\begin{array}{llll}
			\dfrac{\mathrm{d}  S_1}{\mathrm{d}t} & = &\theta(1-p) -\lambda_1(t)S_1(t)-\mu S_1(t), \\\\
			\dfrac{\mathrm{d}S_2}{\mathrm{d}t} & = &\theta p -\lambda_2(t)S_2(t)-\mu S_2(t), \\\\ 
			\dfrac{\mathrm{d}I_1}{\mathrm{d}t} & = &\lambda_1(t)S_1(t)-(\mu + \gamma_1 +d_1)I_1(t), \\\\
			\dfrac{\mathrm{d}I_2}{\mathrm{d}t} & = &\lambda_2(t)S_2(t)-(\mu + \gamma_2 +d_2)I_2(t), \\\\
			\dfrac{\mathrm{d}R}{\mathrm{d}t} & = &\gamma_1I_1(t)+\gamma_2I_2(t)-\mu R(t).
		\end{array}\right.
	\end{equation}
	A summary of the biological significance of the model's parameters (\ref{deterministic}) is given in Table \ref{parameter}.

	\begin{table}[!h]
		\caption{Description and value of the model's parameters.} \label{parameter}
		\centering 
		\begin{tabular}{clccc}
			
			\hline
			\textbf{Parameter} &\textbf{Description} &\textbf{Units}& \textbf{Value}&\textbf{Source }\\\hline
			$\theta$ & Recruitment rate & $person.day^{-1}$&variable &Assumed\\
			$\mu$ & Natural mortality rate &  $day^{-1}$&variable &Assumed\\
			$p$ & Proportion of new hosts vaccinated & - &variable &Assumed\\
			$\beta_{11}$ & Transmission rate form $I_1$ to $S_1$ & $day^{-1}$&variable &Assumed\\
			$\beta_{12}$& Transmission rate form $I_1$ to $S_2$ & $day^{-1}$&variable &Assumed\\
			$\beta_{21}$ & Transmission rate form $I_2$ to $S_1$ & $day^{-1}$&variable &Assumed\\
			$\beta_{22}$ &Transmission rate form $I_2$ to $S_2$  & $day^{-1}$&variable &Assumed\\
			$d_1$ & Mortality rate due to infection of $S_1$   &  $day^{-1}$& $0.0008$  & \cite{Mancusowill2021}\\
			$d_2$ & Mortality rate due to infection of $S_2$  &  $day^{-1}$ & $0.0001$ &\cite{Mancusowill2021}\\
			$\gamma_1$ & Recovery rate of $I_1$ &  $day^{-1}$& $0.1$  &\cite{Mancusowill2021}\\
			$\gamma_2$ & Recovery rate of $I_2$ & $day^{-1}$& $0.13$  &\cite{Mancusowill2021} \\
			\hline
		\end{tabular}
		
	\end{table}
}
\section{Mathematical analysis}
\subsection{Basic properties}
\label{subsec: basicproperties}
{\allowdisplaybreaks
First, we study the basic characteristics of the system solutions: the existence, positivity and  boundedness of solutions. These are 1) essential to make sure that the model (\ref{deterministic}) is well defined mathematically and epidemiologically, and 2) useful for the proofs of the stability results.
}

\subsubsection{Positivity of solutions}
\label{subsubsec: positivityofsolution}
{\allowdisplaybreaks
For any associated Cauchy problem, the system (\ref{deterministic}) which is a $C^{\infty}$-differentiable system, has a unique maximal solution.
	
\begin{thm}\label{positivity}
For any initial condition $(t_0 = 0, X_0=(S_1(0), S_2(0), I_1(0),I_2(0),R(0))\in \mathbb{R}^{5}_{+})$, and for $T\in ]0, +\infty]$, the maximal solution $([0, T[, X=(S_1(t), S_2(t), I_1(t),I_2(t),R(t)))$ of the Cauchy problem associated to the system (\ref{deterministic}) is non-negative.
\end{thm}
\begin{proof}
Let $$\Delta= \{\tilde{t}\in [0, T[\hspace{.2cm}|\hspace{.2cm} S_1(t)>0, S_2(t)>0, I_1(t)>0, I_2(t)>0 \hspace{.2cm}\text{and} \hspace{.2cm}R(t)>0\hspace{.2cm}\forall t\in [0,\tilde{t}[ \}.$$
By continuity of the functions $S_1$, $S_2$, $I_1$, $I_2$ and $R$, one can see that $\Delta \neq \emptyset$. Let $\tilde{T}= \text{sup}\Delta$ be the supremum of  $\Delta$. Now, we want to prove that $\tilde{T}=T$.
		
Assume $\tilde{T} < T,$ then we have that $S_1$, $S_2$, $I_1$, $I_2$ and $R$ are simultaneously  positive on $]0,\tilde{T}[$. Then, at least one of the following conditions is satisfied at time $\tilde{T}$: $S_1(\tilde{T})=0$ and $	\dfrac{\mathrm{dS_1(\tilde{T})}}{\mathrm{d}t}\leqslant0$, $S_2(\tilde{T})=0$ and $\dfrac{\mathrm{dS_2(\tilde{T})}}{\mathrm{d}t}\leqslant0$, $I_1(\tilde{T})=0$ and $\dfrac{\mathrm{dI_1(\tilde{T})}}{\mathrm{d}t}\leqslant0$, $I_2(\tilde{T})=0$ and $\dfrac{\mathrm{dI_2(\tilde{T})}}{\mathrm{d}t}\leqslant$ or $R(\tilde{T})=0$ and $\dfrac{\mathrm{d R(\tilde{T})}}{\mathrm{d}t}\leqslant0$.
		
Suppose that $S_1(\tilde{T})=0$ and $\dfrac{\mathrm{dS_1(\tilde{T})}}{\mathrm{d}t}\leqslant0$,  then we deduce from the first equation of system (\ref{deterministic}) that 
\begin{align*}
\dfrac{\mathrm{d}  S_1(\tilde{T})}{\mathrm{d}t} & = \theta(1-p) -\lambda_1(\tilde{T})S_1(\tilde{T})-\mu S_1(\tilde{T})\\\\	
&= \theta(1-p) > 0
\end{align*}	
which is a contradiction to the previous claim that $\dfrac{\mathrm{dS_1(\tilde{T})}}{\mathrm{d}t}\leqslant0$.
		
We can use the similar argument for all the remaining state variables. Then, $\tilde{T}=T$ and consequently the maximal solution $(S_1(t), S_2(t), I_1(t),I_2(t),R(t))$ of the Cauchy problem related with system (\ref{deterministic}) is positive.
\end{proof}
Therefore, the variables of the system (\ref{deterministic}) are positive for all time $t > 0$. In other terms, solutions of the system (\ref{deterministic}) with non-negative initial conditions will stay positive for all $t > 0$.
	
We now prove a useful lemma.	
}
\subsubsection{Boundedness of solutions}
\label{subsubsec: boundednessofsolutions}
{\allowdisplaybreaks
Since the variables of model (\ref{deterministic}) are non-negative and we are dealing with the dynamic of a number of individuals, it is important and realistic that the total number of individuals does not tend towards infinity. The following result yields:
\begin{lem}
The closed set $$\Omega = \bigg\{(S_1(t), S_2(t), I_1(t),I_2(t),R(t))\in \mathbb{R}^{5}_{+}, \hspace{.2cm}N(t) \leqslant \dfrac{\theta}{\mu} \bigg\}$$ is positively invariant and attracting for the system (\ref{deterministic}).
\end{lem}
\begin{proof}
Using the system (\ref{deterministic}), the dynamics of the total human population satisfy:
$$\dfrac{\mathrm{d}N}{\mathrm{d}t}= \theta -\mu N-d_1I_1-d_2I_2\leqslant \theta-\mu N.$$
		
Integrating both sides of the expression above, we deduce that
		
\begin{equation}\label{invariant}
N(t)\leqslant\dfrac{\theta}{\mu}+\bigg(N(0)-\dfrac{\theta}{\mu}\bigg)e^{-\mu t},\quad \forall t\geqslant 0,
\end{equation}
where $N(0)$ is the value of $N(t)$ at the beginning. 
		
We deduce that if $N(0)\leqslant \dfrac{\theta}{\mu}$, then $0\leqslant N(t)\leqslant \dfrac{\theta}{\mu}$,   $\forall t\geqslant 0$ and $\Omega$ is positively invariant. If $N(0)\geqslant \dfrac{\theta}{\mu}$, then from (\ref{invariant}) the total population decreases and the solutions enters $\Omega$. Hence $N(t)$ is bounded as $t$ \textrightarrow  $\infty$, which means that $\Omega$ is attracting.
\end{proof}
\begin{rem}In finite dimension  every maximal solution of a Cauchy problem is global in a compact set. Then, every maximal solution of the model system (\ref{deterministic}) is global.
\end{rem}
Therefore the solutions of our model are considered epidemiologically and mathematically well posed in $\Omega$ .
	
}

\subsection{Disease-free equilibrium and its stability}
\label{subsec: disease-free}
{\allowdisplaybreaks
For the analysis of the spread of an infection, we define the disease-free equilibrium (DFE) which is a state in the population without any infection. The disease-free equilibrium is deduced from the resolution of the system of equations in (\ref{deterministic}) by taking $I_1 = 0$ and $I_2 = 0$. Thus, the disease-free equilibrium for model (\ref{deterministic}) satisfies the following system of equations:
\begin{equation}\label{DFE}
\left\{ \begin{array}{ll}
			\theta(1-p) -\mu S_1^{0}=0,\\
			\theta p -\mu S_2^{0}=0.
\end{array}\right.
\end{equation}
Solving the system of equations in (\ref{DFE}) yields the disease-free equilibrium point:
$$Q^{0}=(S_1^{0},S_2^{0},0,0,0),$$where  $S_1^{0}=\dfrac{\theta(1-p)}{\mu}$, $S_2^{0}=\dfrac{\theta p}{\mu}$ and $N^{0}=S_1^{0}+S_2^{0}=\dfrac{\theta}{\mu}$.
	
The linear stability of $Q^{0}$ depends on the well known reproduction number $\mathcal{R}_0$, which is defined as the average number of secondary cases caused by an infected individual during its infectivity period when it is introduced into a population of susceptible individuals. We study the stability of the equilibrium through the next generation operator \cite{jacquez1993qualitative,van2002reproduction}. Recalling the notations in \cite{van2002reproduction} for model (\ref{deterministic}), the matrices $\mathcal{F}$ of the new infection and $\mathcal{V}$ of the remaining transfer terms at the DFE for are given by 
$$\mathcal{F} = \begin{bmatrix} \beta_{11}\dfrac{S_1I_1}{N}+\beta_{12}\dfrac{S_1I_2}{N}\\\\ \beta_{21}\dfrac{S_2I_1}{N}+\beta_{22}\dfrac{S_2I_2}{N}\end{bmatrix} \hspace{.2cm}\text{and}\hspace{.2cm}\mathcal{V} =
\begin{bmatrix}(\mu + \gamma_1 +d_1)I_1\\\\ (\mu + \gamma_2 +d_2)I_2\end{bmatrix}.$$
The Jacobian matrices of $\mathcal{F}$ and $\mathcal{V}$ at $Q^{0}$ are respectively,
\begin{equation}\label{FV}
F = \begin{bmatrix} \beta_{11}\dfrac{S_1^{0}}{N^{0}}&\beta_{12}\dfrac{S_1^{0}}{N^{0}}\\\\ \beta_{21}\dfrac{S_2^{0}}{N^{0}}&\beta_{22}\dfrac{S_2^{0}}{N^{0}}\end{bmatrix} \hspace{.2cm}\text{and}\hspace{.2cm}V =
\begin{bmatrix}\mu + \gamma_1 +d_1& 0\\\\ 0& \mu + \gamma_2 +d_2\end{bmatrix}.
\end{equation}
Then, $$FV^{-1}= \begin{bmatrix} \dfrac{\beta_{11}S_1^{0}}{N^{0}(\mu + \gamma_1 +d_1)}&\dfrac{\beta_{12}S_1^{0}}{N^{0}(\mu + \gamma_2 +d_2)}\\\\ \dfrac{\beta_{21}S_2^{0}}{N^{0}(\mu + \gamma_1 +d_1)}&\dfrac{\beta_{22}S_2^{0}}{N^{0}(\mu + \gamma_2 +d_2)}\end{bmatrix},$$ 
and the reproduction number of model system (\ref{deterministic}) is
\begin{align}
\mathcal{R}_0=\rho(FV^{-1})=&\dfrac{1}{2}\Big[ \dfrac{S_1^{0}}{N^{0}}\mathcal{R}_{0,11}+\dfrac{S_2^{0}}{N^{0}}\mathcal{R}_{0,22}+\sqrt{\Big(\dfrac{S_1^{0}}{N^{0}}\mathcal{R}_{0,11}- \dfrac{S_2^{0}}{N^{0}}\mathcal{R}_{0,22}\Big)^{2}+4\dfrac{S_1^{0}}{N^{0}}\dfrac{S_2^{0}}{N^{0}}\mathcal{R}_{0,12}\mathcal{R}_{0,21}}\Big],\nonumber\\
\mathcal{R}_0=&\dfrac{1}{2}\Big[ (1-p)\mathcal{R}_{0,11}+p\mathcal{R}_{0,22}+\sqrt{\Big((1-p)\mathcal{R}_{0,11}- p\mathcal{R}_{0,22}\Big)^{2}+4p(1-p)\mathcal{R}_{0,12}\mathcal{R}_{0,21}}\Big],\label{R_0}
\end{align}
where $\dfrac{S_1^{0}}{N^{0}}=1-p$ (respectively $\dfrac{S_2^{0}}{N^{0}}=p$) is the proportion of susceptible individuals that have not been vaccinated (respectively have been vaccinated) at the DFE $Q^{0}$. Similarly, we define $\mathcal{R}_{0,11}=\dfrac{\beta_{11}}{\mu + \gamma_1 +d_1}$ as the average number of secondary cases generated by an unvaccinated infected individual during its infectious period through the interaction with the unvaccinated population. Also $\mathcal{R}_{0,12}=\dfrac{\beta_{12}}{\mu + \gamma_1 +d_1}$ represents the average number of secondary cases generated by a vaccinated infected in the unvaccinated part of the population, $\mathcal{R}_{0,21}=\dfrac{\beta_{21}}{\mu + \gamma_2 +d_2}$ is the average number of secondary cases generated by an unvaccinated infected in the vaccinated part of the population, and $\mathcal{R}_{0,22}=\dfrac{\beta_{22}}{\mu + \gamma_2 +d_2}$ represents the average number of secondary cases generated by \textcolor{blue}{a} vaccinated infected in the vaccinated part of the population. Further, $\rho(FV^{-1})$ is the spectral radius of $FV^{-1}$. 

\begin{rem}\label{R}From the expression of the reproduction number $\mathcal{R}_0$ in (\ref{R_0}), we deduce that \\$\mathcal{R}_0\geq(1-p)\mathcal{R}_{0,11}\vee p\mathcal{R}_{0,22}$. Moreover using (\ref{R_0}) for $p=0$ (all new hosts are not vaccinated), $\mathcal{R}_0=\mathcal{R}_{0,11}$. Further if $p=1$ (all new hosts are vaccinated), then $\mathcal{R}_0=\mathcal{R}_{0,22}$.
\end{rem}

The importance of the reproduction number is due to the result given in the next lemma derived from Theorem 2 in \cite{van2002reproduction}.
\begin{lem}\label{Re}
The DFE $Q^{0}$ of the system (\ref{deterministic}) is locally asymptotically stable whenever $\mathcal{R}_0< 1$ and unstable whenever $\mathcal{R}_0>1$.
\end{lem}
The biological meaning of Lemma \ref{Re} is that a sufficiently small number of infected hosts does not induce an epidemic unless the reproduction number $\mathcal{R}_0$, is greater than unity. Global asymptotic stability (GAS) of the DFE is required to better control the disease. In addition, the expansion of the basin of attraction of $Q^{0}$ is a more challenging task for the model under consideration, involving a fairly new result. For this purpose, we use Theorems 2.1 and 2.2 \textcolor{blue}{in} \cite{shuai2013global}.
	
\begin{thm}\label{DFEgas}
If $\mathcal{R}_0\leqslant1$, the DFE $Q^{0}$ of the system (\ref{deterministic}) is GAS in $\Omega$. If $\mathcal{R}_0>1$, $Q^{0}$ is unstable, the system (\ref{deterministic}) is uniformly persistent and there exists at least one endemic equilibrium in the interior of $\Omega$.
\end{thm}

\begin{proof}
See \ref{GAS_DFE}.
\end{proof}

As a consequence of the meaning of Theorem \ref{DFEgas} and Remark \ref{R}, we can confidently deduce that the disease can be eradicated from the host community if the value of $\mathcal{R}_0 $ can be reduced to less than the unity, independently of whether individuals introduced in the population are all vaccinated or not.
}

\subsection{Endemic equilibrium and its stability}
\label{subsec: Endemicequilibrium}
{\allowdisplaybreaks
Let $Q^{*}=(S_1^{*}, S_2^{*},I_1^{*},I_2^{*},R^{*})$ be the positive endemic equilibrium (EE) of model system (\ref{deterministic}). Then, the positive endemic equilibrium can be obtained by setting the right hand side of all equations in model system (\ref{deterministic}) to zero, giving:
\begin{equation}\label{EE}
\left\{\begin{array}{llll}
\theta(1-p)-\beta_{11}\dfrac{S_1^{*}I_1^{*}}{N^{*}}-\beta_{12}\dfrac{S_1^{*}I_2^{*}}{N^{*}}-\mu S_1^{*}  = 0,\\\\
\theta p -\beta_{21}\dfrac{S_2^{*}I_1^{*}}{N^{*}}-\beta_{22}\dfrac{S_2^{*}I_2^{*}}{N^{*}}-\mu S_2^{*}  = 0, \\\\ 
\beta_{11}\dfrac{S_1^{*}I_1^{*}}{N^{*}}+\beta_{12}\dfrac{S_1^{*}I_2^{*}}{N^{*}} -(\mu + \gamma_1 +d_1)I_1^{*} = 0, \\\\
\beta_{21}\dfrac{S_2^{*}I_1^{*}}{N^{*}}+\beta_{22}\dfrac{S_2^{*}I_2^{*}}{N^{*}} -(\mu + \gamma_2 +d_2)I_2^{*} = 0, \\\\
\gamma_1I_1^{*}+\gamma_2I_2^{*}-\mu R^{*}  = 0.
\end{array}\right.
\end{equation}
Given the complexity of the system (\ref{EE}), we are not determining an explicit formula for the endemic equilibrium point $Q^{*}$. Note that determining $Q^{*}$ is often very difficult to be carried out when the system is complex and its size is large. However, to prove the existence of $Q^{*}$, we can rewrite the system (\ref{EE}) as a fixed point problem and use Theorem 2.1 in \cite{hethcote1985stability}. To do this, we solve the system (\ref{EE}). After algebraic manipulations of this system, we obtain:
	
$R^{*}=\dfrac{\gamma_1I_1^{*}+\gamma_2I_2^{*}}{\mu}$,  $S_1^{*}=\dfrac{\theta(1-p)N^{*}}{\beta_{11}I_1^{*}+\beta_{12}I_2^{*}+\mu N^{*} }$,  \quad$S_2^{*}=\dfrac{\theta pN^{*}}{\beta_{21}I_1^{*}+\beta_{22}I_2^{*} - d_1I_1^{*} - d_2 I_2^{*}+\theta }$, 
	
$I_1^{*}=\dfrac{\theta(1-p)(\beta_{11}I_1^{*}+\beta_{12}I_2^{*})}{(\mu + \gamma_1 +d_1)(\beta_{11}I_1^{*}+\beta_{12}I_2^{*}- d_1I_1^{*} - d_2 I_2^{*}+\theta  )} =H_1(I^{*})$  and  
	
$I_2^{*}=\dfrac{\theta p(\beta_{21}I_1^{*}+\beta_{22}I_2^{*})}{(\mu + \gamma_2 +d_2)(\beta_{21}I_1^{*}+\beta_{22}I_2^{*}- d_1I_1^{*} - d_2 I_2^{*}+\theta )}=H_2(I^{*})$ with $I^{*}=(I_1^{*},I_2^{*}).$
	
Then, the endemic equilibrium are the fixed points of $H$ given by $I=H(I)$ where $I=(I_1,I_2)$. By definition, $H$ is continuous, monotonously non decreasing and strictly sublinear. $H$ is also a bounded function which maps the non negative orthant $\Omega$ into itself. Morever, $H(0)=0$ by definition and the jacobian of $H$ at the zero, $H^{'}(0)$, exists and is irreducible since $$H^{'}(0)=\begin{bmatrix} \beta_{11}a_1& \beta_{12}a_1\\\\\beta_{21}a_2& \beta_{22}a_2\end{bmatrix}= FV^{-1}, $$ where $a_1=\dfrac{1-p}{\mu + \gamma_1 +d_1}$ and $a_2=\dfrac{p}{\mu + \gamma_2 +d_2}.$

We deduce that the spectral radius $\rho(H^{'}(0))$ of the matrix $H^{'}(0)$ is $\mathcal{R}_0.$  Then, the existence and the uniqueness of a non-negative fixed point occurs if and only if $\mathcal{R}_0>1$.
\begin{prop}
The system (\ref{deterministic}) has only one endemic equilibrium whenever $\mathcal{R}_0>1$.
\end{prop}
We establish the following result to analyze the stability of $Q^{*}$.
	
\begin{thm}\label{EEgas}
If $\mathcal{R}_0>1$, the endemic equilibrium $Q^{*}$ is GAS in $\Omega$.
\end{thm}

\begin{proof}
	See \ref{GAS_EE}.
\end{proof}

The epidemiological consequence of this theorem is that the disease persists as endemic in the host population as soon as $\mathcal{R}_0 > 1$.
	
}
\subsection{Herd immunity threshold}
\label{subsec:Herdimmunitythreshold}
{\allowdisplaybreaks

Herd immunity  is a form of indirect protection from infectious disease that can occur with some diseases when a sufficient percentage of a population has become immune to an infection, whether through previous infections or vaccination, and thereby reducing the likelihood of infection for individuals who lack immunity. This is due to the fact that immune individuals are unlikely to contribute to disease transmission, disrupting chains of infection, which stops or slows the spread of disease. To compute the herd immunity threshold associated with the model (\ref{deterministic}), we set the reproduction number, $\mathcal{R}_0$ to one and
solve for $p=\dfrac{S_2^{0}}{N^{0}}$ which is the proportion of susceptible individuals which have been vaccinated at the DFE, $Q^{0}.$ Then we have,
\begin{align*}
\mathcal{R}_0=1 \Longleftrightarrow& \big[2-\mathcal{R}_{0,11}+(\mathcal{R}_{0,11}-\mathcal{R}_{0,22})p\big]^{2}=\big[\mathcal{R}_{0,11}-(\mathcal{R}_{0,11}+\mathcal{R}_{0,11})p\big]^{2}+ 4p(1-p)\mathcal{R}_{0,12}\mathcal{R}_{0,21} \nonumber\\
\Longleftrightarrow& \big[(\mathcal{R}_{0,11}-\mathcal{R}_{0,22})^{2}-(\mathcal{R}_{0,11}+\mathcal{R}_{0,22})^{2}+4\mathcal{R}_{0,12}\mathcal{R}_{0,21}\big]p^{2}+\big[2(2-\mathcal{R}_{0,11})(\mathcal{R}_{0,11}-\mathcal{R}_{0,22}) \nonumber\\
&+2\mathcal{R}_{0,11}(\mathcal{R}_{0,11}+\mathcal{R}_{0,22})-4\mathcal{R}_{0,12}\mathcal{R}_{0,21}\big]p+(2-\mathcal{R}_{0,11})^{2}-\mathcal{R}_{0,11}^{2}=0. 
\end{align*}
Thus solving $\mathcal{R}_0=1$ is equivalent to finding the roots of polynomial $Q(p)$ given by:
\begin{equation}\label{proportion}
	Q(p)=Ap^{2}+Bp+C,
\end{equation}
where $A=4\mathcal{R}_{0,12}\mathcal{R}_{0,21}-4\mathcal{R}_{0,11}\mathcal{R}_{0,22}$, $B=4\mathcal{R}_{0,11}(1+\mathcal{R}_{0,22})-4(\mathcal{R}_{0,22}+\mathcal{R}_{0,12}\mathcal{R}_{0,21})$ and\\ $C=4(1-\mathcal{R}_{0,11})$.

Noting that negative thresholds are biologically meaningless, the conditions for $Q(p)$ to have positive real roots are determined
below. For this purpose, we now perform a case analysis to determine the positive real zeros of $Q$.

Let $\Delta=B^{2}-4AC$ be the discriminant of the equation $Q(p)=0$.
\begin{description}
\item[Case 1] Suppose $A=0$. Then $$p_{c}=-\frac{C}{B}$$
is the only real root of $Q$. In addition $p_{c}>0$ if and only if $B$ and $C$ have opposite signs and $B\neq0.$
	
\item[Case 2] Suppose $A\neq0$ and $\Delta= 0$. Then
$$p_{c_{0}}=-\dfrac{B}{2A}$$
is the only real root of $Q$. Further $p_{c_0}>0$ if and only if $A$ and $B$ have opposite signs.

\item[Case 3] Suppose $A\neq0$ and $\Delta> 0$. Then
$$p_{c_{1}}=\dfrac{-B-\sqrt{\Delta}}{2A} \hspace{.2 cm} \text{and} \hspace{.2 cm}p_{c_{2}}=\dfrac{-B+\sqrt{\Delta}}{2A} $$
are the real roots of $Q$. 

Moreover, if $A>0$, then 
\begin{equation*}
	\left\{ \begin{array}{ll}
		p_{c_{1}}>0 \hspace{.2 cm} \text{if and only if} \hspace{.2 cm} \sqrt{\Delta}<-B,\\
		p_{c_{2}}>0 \hspace{.2 cm} \text{if and only if} \hspace{.2 cm} \sqrt{\Delta}>B.
	\end{array}\right.
\end{equation*}
Therefore, $Q$ has two positive real roots if $A>0$, $B<0$, $C>0$ and $\Delta> 0$. In addition, it has one positive real root if ($A>0$, $B<0$, $C<0$ and $\Delta> 0$) or ($A>0$, $B>0$ and $C<0$ and $\Delta> 0$).

Finally if $A<0$, then 
\begin{equation*}
	\left\{ \begin{array}{ll}
		p_{c_{1}}>0 \hspace{.2 cm} \text{if and only if} \hspace{.2 cm} \sqrt{\Delta}>-B,\\
		p_{c_{2}}>0 \hspace{.2 cm} \text{if and only if} \hspace{.2 cm} \sqrt{\Delta}<B.
	\end{array}\right.
\end{equation*}
Therefore, $Q$ has two positive real roots if $A< 0$, $B>0$, $C<0$ and $\Delta> 0$. \textcolor{blue}{It} has one positive real root if ($A< 0$, $B>0$, $C>0$ and $\Delta> 0$) or ($A< 0$, $B<0$, $C>0$ and $\Delta> 0$).

\end{description}

Theorem \ref{DFEgas} and Theorem \ref{EEgas} can be combined to give the following result:

\begin{cor} \label{cor}
An imperfect vaccine can lead to the elimination of the disease if $Q(p)>0$ (\textit{i.e.} $\mathcal{R}_0<1$). If $Q(p)<0$ (\textit{i.e.} $\mathcal{R}_0>1$), then the disease persists in the population.
\end{cor}

The implication of Corolloary \ref{cor} is that the use of an imperfect vaccine can lead to the elimination of the disease in the host population, if the proportion of individuals vaccinated satisfies one of these conditions:
\begin{enumerate}
\item $p>p_c$,  if $A=0$, $B>0$ and $C<0$;
\item $p\in [0,p_c[$, if $A=0$, $B>0$ and $C>0$;
\item $p\neq p_{c_{0}}$, if $A>0$, $\Delta=0$ and $B<0$;
\item $p\in [0,p_{c_{1}}[$ or $p>p_{c_{2}}$, if $A > 0$, $\Delta>0$, $B<0$ and $C>0$;
\item $p>p_{c_{1}}$ or  $p>p_{c_{2}}$, if ($A > 0$, $\Delta>0$, $B<0$ and $C<0$) or ($A > 0$, $\Delta>0$, $B>0$ and $C<0$);
\item $p\in ]p_{c_{2}},p_{c_{1}}[$, if $A < 0$, $\Delta>0$, $B>0$ and $C<0$;
\item $p\in [0,p_{c_{1}}[$ or $p\in [0,p_{c_{2}}[$, if ($A < 0$, $\Delta>0$, $B>0$ and $C>0$) or ($A < 0$, $\Delta>0$, $B<0$ and $C>0$).
\end{enumerate}

Conversely, the disease persists in the population if the proportion of individuals vaccinated satisfies one of these conditions:
\begin{enumerate}
	\item $p\in [0,p_c[$,  if $A=0$, $B>0$ and $C<0$;
	\item $p> p_c$, if $A=0$, $B>0$ and $C>0$;
	\item $p\neq p_{c_{0}}$, if $A<0$, $\Delta=0$ and $B>0$;
	\item $p\in ]p_{c_{1}},p_{c_{2}}[$, if $A > 0$, $\Delta>0$, $B<0$ and $C>0$;
	\item $p\in [0,p_{c_{1}}[$ or  $p\in [0,p_{c_{2}}[$, if ($A > 0$, $\Delta>0$, $B<0$ and $C<0$) or ($A > 0$, $\Delta>0$, $B>0$ and $C<0$);
	\item $p\in [0,p_{c_{2}}[$ or $p>p_{c_{1}}$, if $A < 0$, $\Delta>0$, $B>0$ and $C<0$;
	\item $p>p_{c_{1}}$ or $p>p_{c_{2}}$, if ($A < 0$, $\Delta>0$, $B>0$ and $C>0$) or ($A < 0$, $\Delta>0$, $B<0$ and $C>0$).
\end{enumerate}
To conclude on the analytical part, the eradication of a disease is conditioned by the proportion of vaccinated individuals, this vaccination coverage threshold is called the critical vaccination proportion ($p_c$). In some cases, there is one critical proportion which determines whether the basic reproduction number,$\mathcal{R}_0$, is less than one or not. In other cases, two critical proportions are found and which define three different dynamics: disease eradication when $\mathcal{R}_0<1$, endemic disease dynamics when $\mathcal{R}_0>1$ with presence or absence of epidemiological oscillations in the number of infected. In the latter case of two thresholds, the analytical results derive above do not allow the prediction of occurrence of the dynamics and the vaccination proportions. We therefore provide numerical simulations in the follow up section. 

}

\section{Numerical simulations}
We refine the above analytical results by numerical simulations to assess the influence of the various model parameters and the impact of population turnover and trade-offs in vaccination efficiency, on the epidemiological dynamics (\textit{i.e.} the number of infected individuals, and $\mathcal{R}_0$). To illustrate the behavior of model (\ref{deterministic}), we use parameter values for the mortality rates, $d_1$, $d_2$, and the recovery rates, $\gamma_{1}, \gamma_{2}$, measured for Covid-19 as an example of a highly transmissible disease (based on data from the United States \cite{Mancusowill2021}), and vary the values of other parameters as described in Table \ref{parameter}. 

\subsection{Global sensitivity analysis}
Uncertainty / sensitivity analyses are first used to determine which model input parameters have the greatest impact on the epidemiological outcome \cite{Marinometh2008}. The sensitivity analysis of the model parameters is carried out to measure the correlation between the model’s parameters (\ref{deterministic}) and 1) the total number of infected individuals ($I_1+I_2$), and 2) the threshold parameter $\mathcal{R}_0$. The analysis is performed by using the Latin Hypercube Sampling (LHS) technique and partial rank correlation coefficients (PRCCs) \cite{Marinometh2008}. In our analysis, 1,000 model simulations are performed by running the model for 200 time steps (equivalent to 200 days) and number of infected are recorded at time points 50, 100 and 200. To perform the sensitivity analysis, each parameter has a parameter range defines by the maximum (respectively the minimum) being $50\%$ greater (respectively less) than its baseline (values in Table \ref{sensibilité1},      
\ref{sensibilité2}, \ref{sensibilité3}, \ref{sensibilité4}). We then divide each parameter range into 1,000 equally large sub-intervals, and draw a value per parameter within that interval using a Uniform draw. By this mean we obtain a uniform distribution of 1,000 parameter values for each parameter. The parameter space (or LHS matrix) has dimension of length 11 with each dimension specifying an uncertain parameter vector of length 1,000. The base parameter values are chosen to define several scenarios of interest regarding the intensity of the turnover (weak and strong) and efficiency of the vaccine (weak and strong). In PRCC analysis, the parameters with the larger positive or negative PRCC values ($> 0.5$ or $< -0.5$) and with corresponding small p-values ($< 0.05$) are deemed the most influential in determining the outcome of the model. A positive (negative) correlation coefficient corresponds to an increasing (decreasing) monotonic trend between the chosen response function and the parameter under consideration. The results of the PRCC analyses are found in Tables \ref{sensibilité1}, \ref{sensibilité2}, \ref{sensibilité3}, \ref{sensibilité4} in \ref{Tables} .

\label{subsec:sensitivityanalysis}
{\allowdisplaybreaks
\begin{table}
	\caption{Summary of the influence of parameters on the total numbers of infected at different time points. } \label{sensibilité5}\scalebox{0.85}{\begin{tabular}{c} 
	\begin{tabular}{c|ccc} \hline
			Scenarios &  \multicolumn{3}{c}{ Total Infected: $I_1+I_2$}\\
			\cline{2-4}  & $t=50$ days&$t=100$ days& $t=200$ days\\
			\hline 	
		Strong turnover and weak efficiency &$\theta$(+), $\beta_{11}(+)$, $\mu(-)$,$\gamma_{1}(-)$&$\theta(+)$, $\beta_{11}(+)$, $\mu(-)$,$\gamma_{1}(-)$&$\theta(+)$, $\beta_{11}(+)$, $\mu(-)$,$\gamma_{1}(-)$\\
		Strong turnover and strong efficiency& $\theta(+)$, $\beta_{11}(+)$,$\mu(-)$,$\gamma_{1}(-)$&$\theta(+)$, $\beta_{11}(+)$, $\mu(-)$,$\gamma_{1}(-)$&$\theta(+)$, $\beta_{11}(+)$, $\mu(-)$,$\gamma_{1}(-)$ \\
		Weak turnover and weak efficiency & $\beta_{11}(-)$,$\beta_{21}(-)$, $\beta_{22}(-)$&$\beta_{21}(-)$, $\beta_{22}(-)$,$\gamma_{1}(+)$,$\gamma_{2}(+)$&$\theta(+)$, $\beta_{21}(-)$,$\gamma_{1}(+)$\\
		Weak turnover and strong efficiency&$\theta(+)$, $\beta_{11}(-)$,$\beta_{21}(-)$,$\gamma_{1}(-)$&$\beta_{21}(-)$,$\mu(-)$,$\gamma_{1}(+)$&$\theta(+)$, $\beta_{21}(-)$,$\mu(-)$,$\gamma_{1}(+)$\\
			\hline
		\end{tabular} 
	\end{tabular}}
\end{table}

Based on the results from Tables \ref{sensibilité1}, \ref{sensibilité2}, \ref{sensibilité3}, \ref{sensibilité4}, we provide in table \ref{sensibilité5}, a summary of the the parameters that significantly affect the number of infected. Overall, it appears that the recruitment rate, $\theta$ and the recovery rate of the infected who have not been vaccinated, $\gamma_1$, are the two main parameters driving the number of infected. This suggests that an effective control strategy should aim to limit significantly the immigration of new hosts in the population (to decrease $\theta$) and improve the treatment of infected people (to increase $\gamma_1$). We then proceed to a similar analysis with $\mathcal{R}_0$, and summarize the sensitivity analysis of the LHS and PRCC techniques in Figure \ref{PrccR}. We find, perhaps unsurprisingly, that the proportion of new hosts vaccinated, $p$, is the most significant parameter explaining the change in $\mathcal{R}_0$, along with the transmission rate from unvaccinated infected to unvaccinated susceptibles, $\beta_{11}$ and the recovery rate of the infected who have not been vaccinated, $\gamma_1$ (Table \ref{sensibilité5}). 

\begin{figure}[!h]
	\centering
	\begin{subfigure}{0.45\textwidth}
		\includegraphics[width=1.1\textwidth]{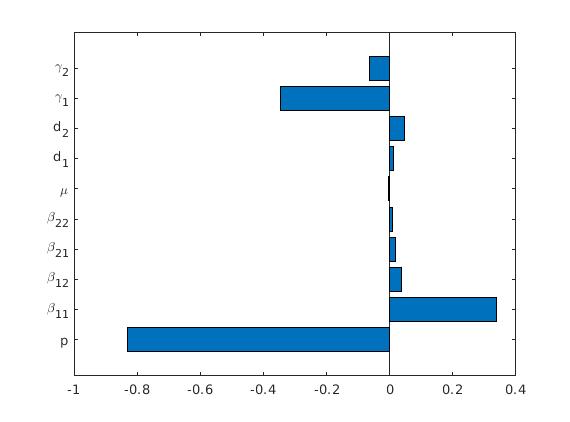}
		\caption{Strong turnover and weak efficiency}
	\end{subfigure}
	\begin{subfigure}{0.45\textwidth}
		\includegraphics[width=1.1\textwidth]{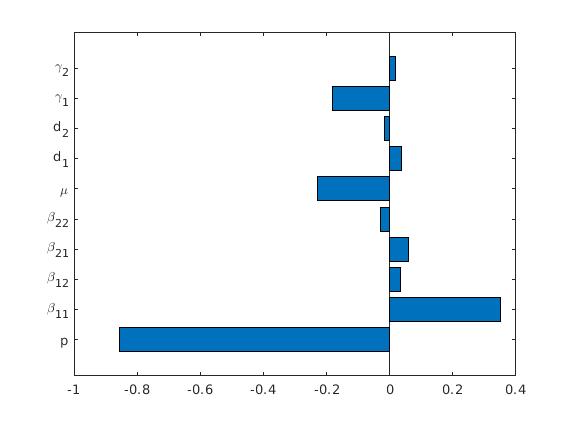}
		\caption{Strong turnover and strong efficiency}
	\end{subfigure}
	\newline
	
	\begin{subfigure}{0.45\textwidth}
		\includegraphics[width=1.1\textwidth]{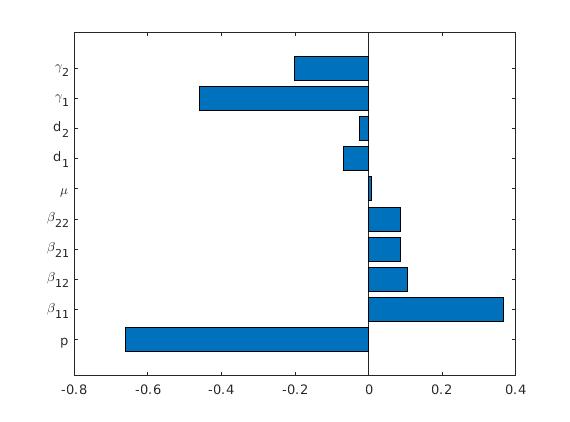}
		\caption{Weak turnover and weak efficiency}
	\end{subfigure}
	\begin{subfigure}{0.45\textwidth}
		\includegraphics[width=1.1\textwidth]{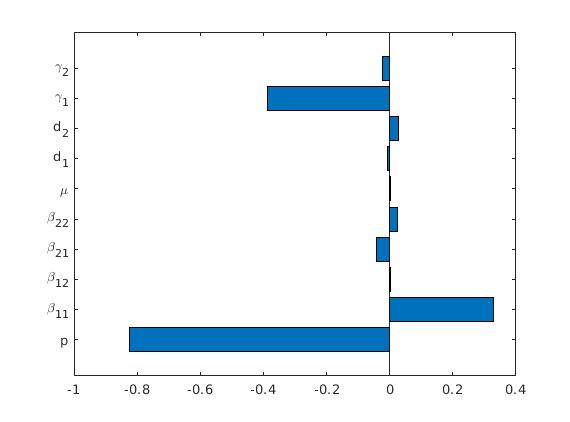}
		\caption{Weak turnover and strong efficiency}
	\end{subfigure}
	\caption{PRCCs describing the impact of model's parameters on $\mathcal{R}_0$ of the model (\ref{deterministic}) with respect to some scenarios. The range of the parameters in (a) (respectively in (b), (c) and (d)) is the same as given on Table \ref{sensibilité1}(respectively on Table \ref{sensibilité2},\ref{sensibilité3},\ref{sensibilité4}).}
	\label{PrccR}
\end{figure}

}

\subsection{Interplay between vaccine efficiency and population turnover}
\label{subsec:Impactofthevaccineefficiency}
{\allowdisplaybreaks
We now study the effect of population turn-over and vaccine efficiency on the epidemiological dynamics. Specifically, we use numerical simulations to find the vaccination coverage necessary to eradicate the disease in the community ($\mathcal{R}_0$ satisfying the corollary \ref{cor}) under two population turnover rates (fixing the ratio $\theta / \mu$, we define strong turnover with $\theta=1000$ and $\mu=0.09$, and weak with $\theta=10$ and $\mu=0.0009$), when the efficiency of the vaccine only reduces transmission. The vaccine efficiency is set as weak ($\beta_{21}=(1-0.5)\beta_{11}$ and $\beta_{22}=(1-0.5)\beta_{12}$, defining an efficiency of $50\%$) or strong ($\beta_{21}=(1-0.9)\beta_{11}$ and $\beta_{22}=(1-0.9)\beta_{12}$, defining an efficiency of $90\%$).
	
	\begin{figure}[!h]
	\centering
	\begin{subfigure}[b]{0.45\textwidth}
		\includegraphics[width=1\linewidth]{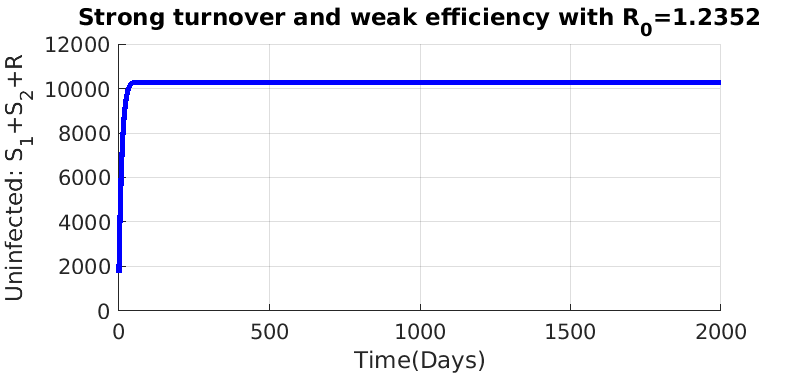}
		\caption{}
	\end{subfigure}
	\begin{subfigure}[b]{0.45\textwidth}
		\includegraphics[width=1\linewidth]{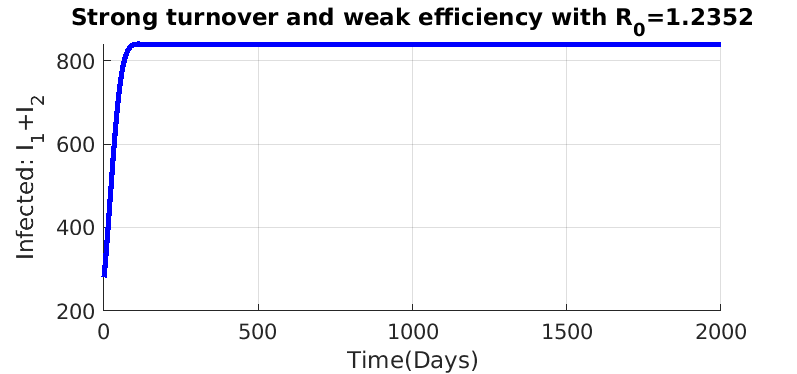}
		\caption{}\label{STWE}
	\end{subfigure}	
\begin{subfigure}[b]{0.45\textwidth}
	\includegraphics[width=1\linewidth]{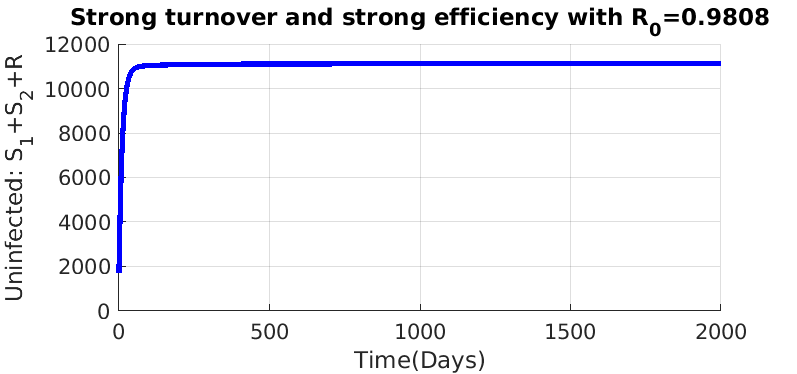}
	\caption{}
\end{subfigure}
\begin{subfigure}[b]{0.45\textwidth}
	\includegraphics[width=1\linewidth]{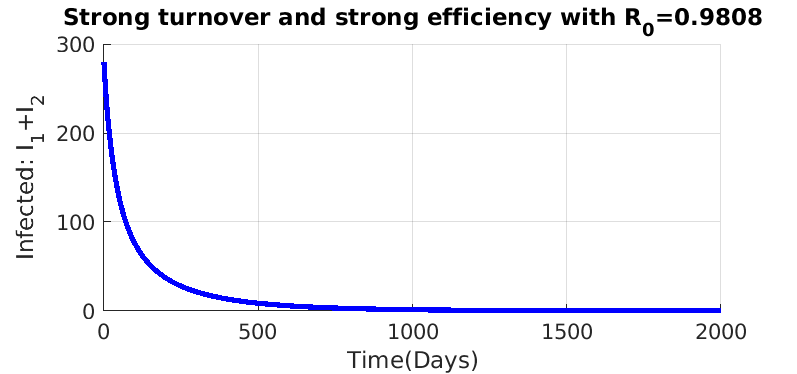}
	\caption{}\label{STSE}
\end{subfigure}
	\caption{Epidemiological dynamics with the initial conditions $S_1(0)=1000$, $S_2(0)=700$, $I_1(0)=200$, $I_2(0)=80$, $R(0)=20$ for various scenarios assuming the parameters $\beta_{11}= 0.35$, $\beta_{12}=0.28$, $p=0.5$ and strong population turnover ($\theta=1000$, $\mu=0.09$). We present under weak vaccine efficiency ($\beta_{21}=0.175, \beta_{22}=0.14$), the number of (a) uninfected and (b) infected individuals. We present under strong vaccine efficiency ($\beta_{21}=0.035, \beta_{22}=0.028$) the number of (c) uninfected and (d) infected individuals. Others parameters values are as in Table \ref{parameter}.}		
\end{figure}	
	
\subsubsection{Strong population turnover}
The epidemiological dynamics in Figure\ref{STWE} under strong turnover and weak vaccine efficiency ($\mathcal{R}_0=1.2352$) shows that the dynamics reaches the endemic disease equilibrium. Furthermore if $p$ takes value between $0$ and $p_1$ (with $p_1\approx0.696$), the basic reproduction number is greater than $1$, but if $p$ is between $p_1$ and $1$, the basic reproduction number is less than $1$ (as predicted in the analytical results in Corollary \ref{cor}). So to eradicate the disease under strong population turnover and weak efficiency of the vaccine, a minimum vaccination rate is needed and defined by $p_1$. Under strong turnover and strong efficiency (Figure\ref{STSE}, with $\mathcal{R}_0=0.9808$) the disease becomes extinct. Furthermore if the parameter $p$ between $0$ and $p_2$ with $p_2\approx0.489$, the basic reproduction number is greater than $1$, while for $p$ between $p_2$ and $1$, the basic reproduction number is less than $1$. So to eradicate the disease in this context of strong turnover and strong efficiency of the vaccine, there is a need to vaccinate more than $48.9\%$ of the new host individuals.
	
	\begin{figure}[!h]
		\centering
		\begin{subfigure}[b]{0.45\textwidth}
			\includegraphics[width=1\linewidth]{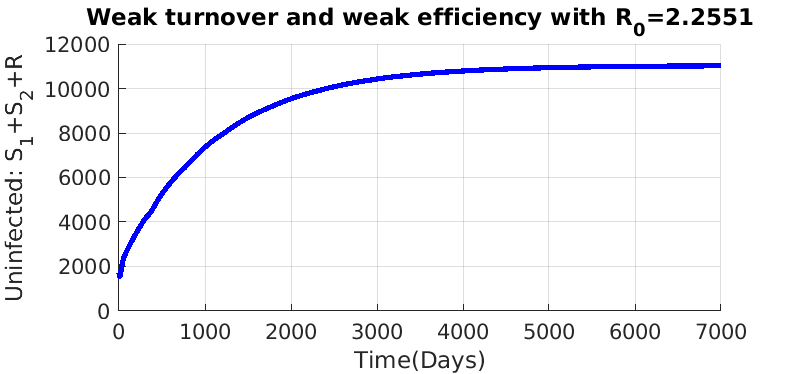}
			\caption{}
		\end{subfigure}
		\begin{subfigure}[b]{0.45\textwidth}
			\includegraphics[width=1\linewidth]{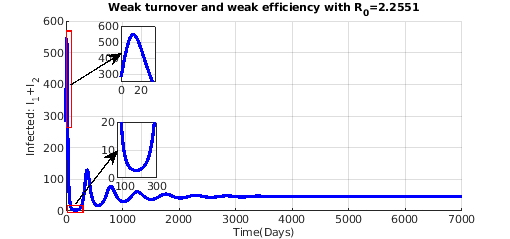}
			\caption{}\label{WTWE}
		\end{subfigure}
		\begin{subfigure}[b]{0.45\textwidth}
		\includegraphics[width=1\linewidth]{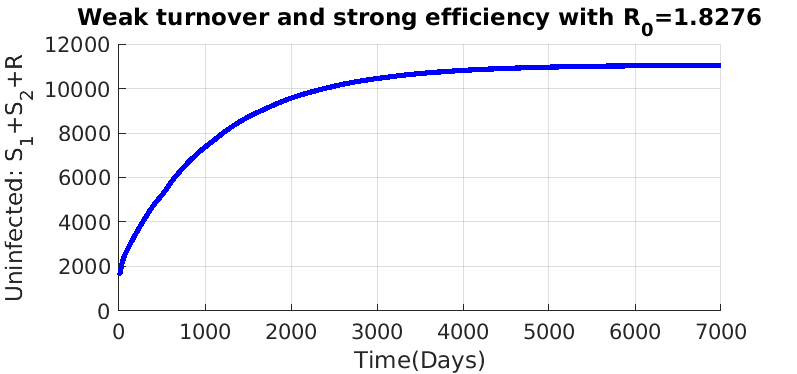}
		\caption{}
	\end{subfigure}
	\begin{subfigure}[b]{0.45\textwidth}
	\includegraphics[width=1\linewidth]{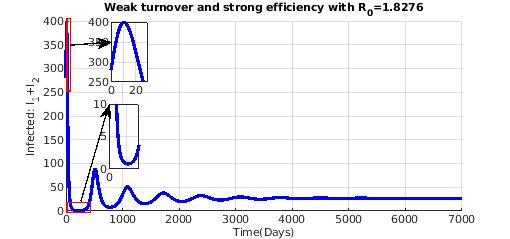}
	\caption{}\label{WTSE}
\end{subfigure}
		\caption{Simulation of model (\ref{deterministic}) at the initial conditions $S_1(0)=1000$, $S_2(0)=700$, $I_1(0)=200$, $I_2(0)=80$, $R(0)=20$ when $\theta=10$, $\beta_{11}= 0.35$, $\beta_{12}=0.28$, $\beta_{21}=0.175, \beta_{22}=0.14$, $\mu=0.0009$, $p=0.5$, (a) Uninfected individuals in weak turnover and weak efficiency scenario and (b) Infected individuals in weak turnover and weak efficiency scenario. When $\theta=1000$, $\beta_{11}= 0.35$, $\beta_{12}=0.28$, $\beta_{21}=0.035, \beta_{22}=0.028$, $\mu=0.0009$, $p=0.5$,(c) Uninfected individuals in weak turnover and strong efficiency scenario and (d) Infected individuals in weak turnover and strong efficiency scenario. Others parameters values are as in Table \ref{parameter}.}
	\end{figure}

\subsubsection{Weak population turnover}
To illustrate a weak population turnover, we consider the values $\theta=10$ and $\mu=0.0009$, noting that  the ratio of $\theta / \mu$ is the same as for the strong turnover investigated above. Under weak turnover, the epidemiological dynamics exhibits damped oscillations (recurring outbreaks) before stabilizing at the endemic state with disease persistence (Figure\ref{WTWE} with $\mathcal{R}_0=2.2551$, Figure\ref{WTSE} with $\mathcal{R}_0=1.8276$). 
These oscillations are due to the fact that individuals migrate rapidly in the recovered compartment, and a new outbreak only occurs when a sufficient number of susceptible are available from new recruitment into the population and recovered individuals loosing their immunity (so-called waning immunity). This phenomenon was also described in \cite{Ashby2021, Pulliam2007, Gumel2006, McLean2002}, and the effect of turnover and waning immunity is specifically described in \cite{Ashby2021, Pulliam2007}. \\
With respect to the control of the disease, under weak vaccine efficiency, $p$ can take any value between $0$ and $1$, the basic reproduction number is always greater than $1$ (Figure\ref{WTWE} with $\mathcal{R}_0=2.2551$). In contrast, when vaccine efficiency is strong, three cases occur Figure\ref{WTSE} (with $\mathcal{R}_0=1.8276$). When $p$ has a value between $0$ and $p_3$ with $p_3\approx0.753$, the basic reproduction number is greater than $1$ and we observe a damped periodicity of the number of infected individuals converging towards a stable endemic state. When $p$ takes values between $p_3$ and $p_4$ (with $p_4\approx0.756$), the basic reproduction number, $\mathcal{R}_0$, is greater than $1$ but there are no periodic oscillations. And for $p \in [p_4, 1]$, the basic reproduction number, $\mathcal{R}_0$, is less than $1$, and disease becomes extinct. Note that between $p_3$ and $p_4$, the behavior can change very finely, but the resolution of our simulations does not allow us to decide on a very precise bound when oscillations occur or not. Therefore, to eradicate the disease in this context of weak population turnover and strong efficiency of the vaccine, a high vaccination coverage (more than $75.6\%$ of the new host individuals) is needed. Our results extend those in \cite{Nuismer2016} showing that it is feasible to control disease by a weakly efficient vaccine acting on disease transmission, but that the required vaccination coverage depends on the population turnover. We note that the persistence of an endemic equilibrium is predicted by the condition $\mathcal{R}_0>1$, even if damped oscillations in the number of infected individuals occur. In other words, while the population turnover does not factor directly in the analytical expression of $\mathcal{R}_0$, it enters only indirectly by affecting the proportion of susceptible individuals available (eq. \ref{R_0}). The simulation results provide examples of the analytical expressions obtained in eq.~\ref{proportion} following the Corollary \ref{cor}.

\subsection{Interplay between types of vaccines and population turnover}
\label{subsec:Assessmentofcombiningdifferenttypesofvaccines}
{\allowdisplaybreaks	
We now assume that a vaccine has two potential mechanisms of action on the disease, namely blocking transmission and/or favouring the recovery of infected individuals. We investigate the effect of these vaccine types on the epidemiology depending on the population turnover. Specifically, model (\ref{deterministic}) is slightly modified to allow for the assessment of the efficiency of the vaccine regarding  the probability of being infected and the recovery rate. This is achieved by simply rescaling the parameters as follows: 
\begin{align}
\beta_{21}=(1-\varepsilon)\beta_{11},\beta_{22}=(1-\varepsilon)\beta_{12}, \hspace{.1cm} \text{and} \hspace{.1cm} \gamma_{1}=(1-\nu)\gamma_{2},   \label{scale}
\end{align}
where $0\leq\varepsilon\leq 1$ represents the effect of the vaccine on the transmission and $0\leq\nu\leq 1$ represents the effect of the vaccine on the ability of being recovered. Substituting the rescaled expressions
for (\ref{scale}) into the model (\ref{deterministic}), one deduces that the basic reproduction number the model (\ref{deterministic}) can be rewritten as:
	
\begin{align}
	\mathcal{R}_0=&\dfrac{1}{2}\Big[ (1-p)\mathcal{R}_{0,11}+p\mathcal{R}_{0,22}+\sqrt{\Big((1-p)\mathcal{R}_{0,11}- p\mathcal{R}_{0,22}\Big)^{2}+4p(1-p)\mathcal{R}_{0,12}\mathcal{R}_{0,21}}\Big],\label{R_0m}
\end{align}
with $\mathcal{R}_{0,11}=\dfrac{\beta_{11}}{\mu + (1-\nu)\gamma_{2} +d_1}$,  $\mathcal{R}_{0,12}=\dfrac{\beta_{12}}{\mu + (1-\nu)\gamma_{2} +d_1}$, $\mathcal{R}_{0,21}=\dfrac{(1-\varepsilon)\beta_{11}}{\mu + \gamma_{2} +d_2}$ and $\mathcal{R}_{0,22}=\dfrac{(1-\varepsilon)\beta_{12}}{\mu + \gamma_{2} +d_2}$. Simulations are carried out to assess the interplay of the type of vaccine and the population turnover.

\begin{figure}[!h]
	\centering
	\begin{subfigure}[b]{0.48\textwidth}
		\includegraphics[width=1.1\linewidth]{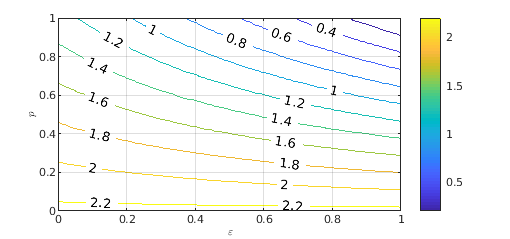}
		\caption{}\label{ContourST1}
	\end{subfigure}
	\begin{subfigure}[b]{0.48\textwidth}
		\includegraphics[width=1.1\linewidth]{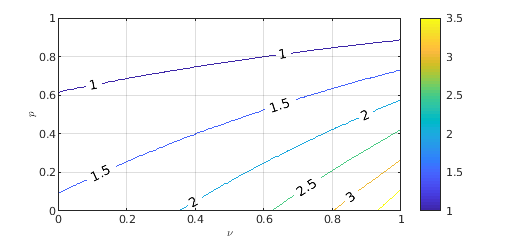}
		\caption{}\label{ContourST2}
	\end{subfigure}
	
	\begin{subfigure}[b]{0.48\textwidth}
		\includegraphics[width=1.1\linewidth]{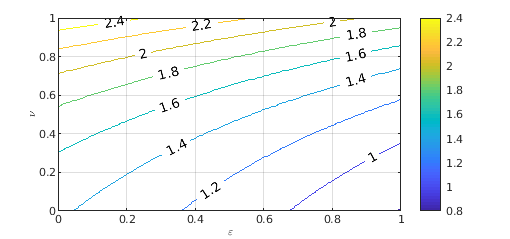}
		\caption{}\label{ContourST3}
	\end{subfigure}	
	\caption{Contour plots of the basic reproduction number ($\mathcal{R}_0$) of the model (\ref{deterministic}) with a strong population turnover as a function of (a) vaccination coverage, $p$, and vaccine efficiency on disease transmission, $\varepsilon$ (with fixed $\nu=0.5$); (b) vaccination coverage, $p$, and vaccine efficiency on recovery, $\nu$ (with fixed $\varepsilon=0.5$); and (c) vaccine efficiency on recovery, $\nu$, and vaccine efficiency on transmission, $\varepsilon$ (with fixed $p=0.5$). The parameters are $\theta=1000$, $\beta_{11}= 0.35$, $\beta_{12}=0.28$, $\beta_{21}=0.175, \beta_{22}=0.14$, $\mu=0.09$, $d_1=0.0008$, $d_2=0.0001$, $\gamma_{1}=0.065$, $\gamma_{2}=0.13$.}\label{ContourST}
\end{figure}

Under a strong population turnover, as expected, the value of the reproduction number decreases as coverage and efficiency of the vaccine on the transmission increase (Figure \ref{ContourST1}), and if the vaccine is designed to only decrease the transmission by $80\%$ (\textit{i.e.} $\varepsilon=0.8$), the eradication of the disease in the host population can be achieved ($\mathcal{R}_0 <1$) if at least $70\%$ of the population is vaccinated (Figure \ref{ContourST1}). On the other hand, the value of the reproduction number decreases as coverage increases and efficiency of the vaccine favoring recovery decreases (Figure \ref{ContourST2}). With a vaccine designed to enhance recovery by $20\%$ (\textit{i.e.} $\nu=0.2$), the eradication of the disease in the host population can be achieved ($\mathcal{R}_0 <1$) if at least $68\%$ of the population is vaccinated (Figure \ref{ContourST2}). In Figure \ref{ContourST3}, we present the effect of the combined efficiency of the vaccine (decreasing transmission and favouring recovery) on the reproduction number at $p=0.5$. The eradication of the disease can be achieved ($\mathcal{R}_0<1$) if the vaccine has a combined efficiency of at least $85\%$ against infection (and thus transmission) and at least $20\%$ to enhance recovery (for a given vaccination coverage of $p=0.5$). These figures represent subsets of the general results presented in Figure \ref{ScatterST}, in which $\mathcal{R}_0$ is a function of $\varepsilon$, $\nu$ and $p$. The use of a vaccine with a combined efficiency (decreasing transmission and favouring recovery) can be associated to the vaccination coverage in order to achieve the elimination of the disease. For example, with a vaccination coverage of $20\%$ ($p=0.2$), it is not possible to eliminate the disease no matter the combined efficiency of the vaccine (Figure \ref{SliceST}), while at $80\%$ coverage ($p=0.8$), there are several combinations of vaccine types, decreasing transmission and favouring recovery, that can promote disease control (Figure \ref{SliceST}). \\

The above results change dramatically under a weak population turnover. As expected, the value of the reproduction number decreases as coverage and efficiency of the vaccine on the transmission increase (Figure \ref{ContourWT1}), but a higher vaccination coverage is needed compared to the strong population turnover to achieve $\mathcal{R}_0<1$. Moreover, it is not possible to eradicate the disease if 1) the vaccine is only efficient to enhance recovery, no matter the vaccination coverage (Figure \ref{ContourWT2}), or 2) if the efficiency of the vaccine is combined but vaccination coverage is $p=0.5$ (Figure \ref{SliceWT}). The general results of $\mathcal{R}_0$ as a function of $\varepsilon$, $\nu$ and $p$ demonstrate that under weak population turnover, disease eradication requires a very strong efficiency of the vaccine and a high coverage (Figure \ref{ScatterWT}).
	
}

\subsection{Interplay between vaccine efficiency trade-off and population turnover}
\label{subsec:Assessmentofvaccineefficiencytrade-off}
{\allowdisplaybreaks

So far we have assumed that all parameters of vaccine efficiency can be independently chosen from one another. We study, here, the epidemiological dynamics when there exists a possible (and realistic) trade-off (relationship) between the vaccine efficiency on the transmission and on the recovery. We assume three possible trade-off curves: convex($\nu=\varepsilon^{2}$), concave($\nu=\sqrt{\varepsilon}$) or linear($\nu=\varepsilon$). Under a strong population turnover, assuming a vaccine of at least $60\%$ of efficiency, disease eradication can be achieved ($\mathcal{R}_0<1$) if the coverage is at least $65\%$ under a convex trade-off (Figure \ref{contourconveST}), at least $80\%$ under a concave trade-off (Figure \ref{contourconcaST}) and at least $75\%$ under a linear trade-off (Figure \ref{contourlineaST}). Imposing vaccine trade-off affects therefore the shape of the $\mathcal{R}_0$ curves in Figure \ref{contourconveST}, \ref{contourconcaST}, \ref{contourlineaST} compared to Figures \ref{ContourST1} and \ref{ContourST2}, and may be important to predict the minimum vaccination coverage to be achieved. However under a weak population turnover, the disease persists no matter the vaccination coverage and whatever trade-off are assumed in the vaccine (Figures \ref{contourconveWT}, \ref{contourconcaWT} and \ref{contourlineaWT}).

\begin{figure}[!h]
	\centering
	\begin{subfigure}[b]{0.48\textwidth}
		\includegraphics[width=1.1\linewidth]{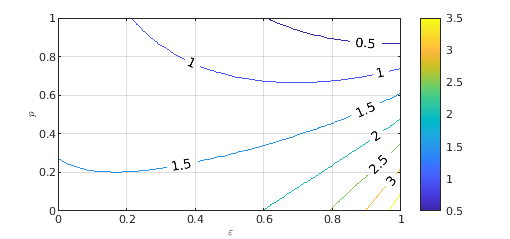}
		\caption{}\label{contourconveST}
	\end{subfigure}
	\begin{subfigure}[b]{0.48\textwidth}
		\includegraphics[width=1.1\linewidth]{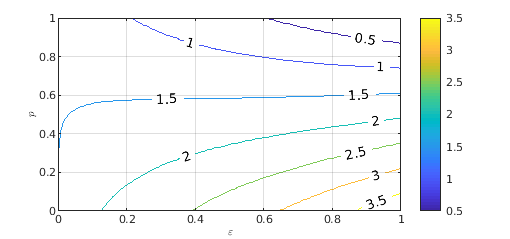}
		\caption{}\label{contourconcaST}
	\end{subfigure}
	
	\begin{subfigure}[b]{0.48\textwidth}
		\includegraphics[width=1.1\linewidth]{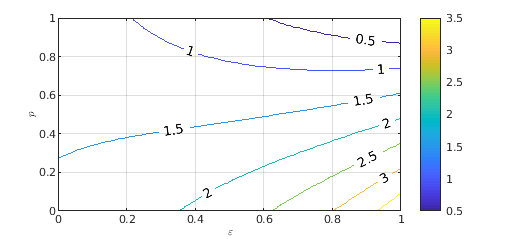}
		\caption{}\label{contourlineaST}
	\end{subfigure}	
	\caption{Contour plots of the basic reproduction number ($\mathcal{R}_0$) of the model (\ref{deterministic}) with a strong population turnover as a function of vaccine coverage, $p$, and vaccine efficiency on the transmission, $\varepsilon$ when: (a)  $\nu=\varepsilon^{2}$ (convex relationship); (b) $\nu=\sqrt{\varepsilon}$ (concave relationship); (c)$\nu=\varepsilon$ (linear relationship). The parameters are $\theta=1000$, $\beta_{11}= 0.35$, $\beta_{12}=0.28$, $\beta_{21}=0.175, \beta_{22}=0.14$, $\mu=0.09$, $d_1=0.0008$, $d_2=0.0001$, $\gamma_{1}=0.065$, $\gamma_{2}=0.13$.}
\end{figure}

}

\section{Discussion and Conclusion}
\label{sec:conclusion}
When a large proportion of a population becomes immune to a virus, it becomes harder for the disease to spread. This is the core concept underlying the concept of herd immunity \cite{Djatcha2017,Ashby2021, Mancusowill2021}. However, there are numerous individuals who refuse to be vaccinated because of various reasons (health concerns, lack of information, systemic mistrust, ...), and some vaccines provide only partial protection from disease or can be only efficient against few disease variants (see the recent Covid-19 epidemics and the vaccine efficiency and waning of immunity against different variants). Therefore, it is rather common that pathogens face an heterogeneous population of vaccinated and unvaccinated hosts, and this has consequences for the evolution of the disease itself \cite{gandon2003imperfect, Alizon2009, gandon2007evolutionary}. In this study, we used mathematical modelling
approaches (analysis and numerical simulations) to assess the potential population-level impact of the using different types of imperfect vaccines to  control the burden of a disease in a community. In a first part, we provide a theoretical analysis of the model, including the basic reproduction number $\mathcal{R}_{0}$ and conditions for the stability of the equilibria. We derive the condition to be satisfied regarding the proportion of vaccinated individuals at steady state in order to attain herd immunity. We express this condition as the critical coverage to be achieved for $\mathcal{R}_0<1$. \\

When the vaccine is developed to prevent infection and stop transmission, our result based on Covid-19 parameter estimates show that it is possible to eliminate the disease with a strong population turnover if the vaccination coverage is greater than $69.6$\% (respectively $48.9$\%) with a  weak (respectively strong) efficiency of the vaccine. However, when population turnover is weak, we observe damped oscillations and eradication is possible with a vaccine with high efficiency and a coverage greater than $75.6$\%. Otherwise, the disease persists and becomes endemic in the community. We highlight here the effect of population turnover as an important first factor in deciding the effectiveness of vaccination campaigns (as suggested in \cite{Scherer2002, Pulliam2007, Knight2020}). For example of application to a human population, the turnover can be consider as migration in and out of the community since the birth and death rate are usually small and fairly constant. Our results suggest that for a community with strong migration (strong turnover), we can vaccinate individuals coming in in order to reduce the basic reproduction number. However, if there is a weak migration (weak turnover) as for example occurred during lockdown when flights and travel are restricted, the vaccination strategy should be improved by undertaking a mass vaccination campaign and using a high efficiency vaccine. A similar reasoning applies to livestocks with (potential vaccinated) calf migrations between farms which influences the epidemic.  }

We then complexify our analysis to analyse more finely the effect of the type of vaccine and its efficiency on disease dynamics. The vaccine can decrease transmission and/or favour recovery of infected individuals. Disease eradication is possible if the vaccine decreases transmission by $82$\%,  enhances recovery by at least $25$\% and a vaccination coverage of $82$\% is achieved under a strong population turnover. Under weak turnover, maximum vaccine efficiency and coverage are required. Therefore, there is also an interplay between the strength of population turnover and the efficiency of the vaccine (and the property of the vaccine). Finally, we explore the importance of vaccine design if trade-off between the vaccine efficiency to stop transmission (infection) and disease recovery are expected. We use three trade-off curves, and show that the convex ($\nu = \varepsilon^2$) function is the most desirable, when the efficiency of the vaccine is at least $60$\% under a strong turnover of population. However, under a weak population turnover, the disease cannot be easily eradicated no matter the vaccination coverage and the efficiency of a combined vaccine. Furthermore, we notice that a smaller vaccination coverage and/or efficiency is needed when using a vaccine designed with a convex trade-off between the above two properties (decrease transmission and favour recovery) than other vaccines (different trade-offs or no trade-off). 

Our model has some limitations and advantages compared to previous work in the literature. First, we use, for illustrative purpose, Covid-19 parameters to exemplify expected threshold for vaccination coverage for a highly transmissible disease. Second, our model does not explicitly account for a continuous vaccination (or a large vaccination campaigns) of individuals in a community. Vaccination is linked in our model to the population turnover, explaining the appearance of periodic oscillations in disease incidence (the honey moon periods). Such periodic epidemics occur and are predicted for Covid-19, but as a consequence of immunity waning of the various vaccines against new variants \cite{Mancusowill2021}. Third, we use a frequency-dependent transmission which allows us to derive analytical results in more depth than some previous models, but may underestimate the spread of disease and speed of disease dynamics.

This model contains some general conlusions which are not only applicable to human populations, but also domesticated animals or even crops. Domesticated animals also require vaccinations, and our study draws recommendations on the importance of turnover and migration rates in and out of the population. Our results also suggest that in domesticated animals, the type of vaccine can be adjusted depending on the disease, especially if it is desirable that infected animal recover well, rather than attempting to prevent any transmission. In addition, we also suggest that the principles of the model apply to plant (crop) immunization. To protect plants against invasion of pathogens or pests, one can use different biotic and synthetic chemicals to induce immunity in the plant \cite{Dyakov2007} or protect plants by spraying fungicides. In a field, or among fields, some plants will be more resistant than others for a certain period of time. The spray is equivalent to the vaccination, and is in that case decoupled from the population turnover which is the planting/renewal and harvesting/removal of plants. Plant epidemiology modelling has been used to predict the efficiency of imperfect fungicide treatments on the epidemics and on yield \cite{Rock2014}, with results mirroring our own.

In summary, our study showed that it is possible to achieve disease control by vaccination in a population with strong turnover, even if we use a weak imperfect vaccine designed to reduce only transmission. However, a higher vaccination coverage and a strong efficiency vaccine are necessary to control the disease under weak population turnover. Besides, a vaccine with convex trade-off between the efficiency to reduce transmission and to enhance recovery is recommendable along with a high vaccination coverage.

\section*{Authors contributions}
Conception and design: HLNB, OMP, AT; Formal investigation: HLNB, OMP; Numerical simulations: HLNB; Writing first draft: HLNB; Supervision: OMP, AT, JMN; Revision of draft: AT, OMP, JMN.

\section*{Declaration of Competing Interest}
\noindent The authors declare that they have no known competing financial interests or personal relationships that could have appeared to influence the work reported in this paper.

\section*{Acknowledgements}
HLB was funded by a grant from the African Institute for Mathematical Sciences,
www.nexteinstein.org, with financial support from the Government of Canada, provided through Global Affairs Canada,\\ www.international.gc.ca, and the International Development Research Centre, www.idrc.ca. OPM acknowledge the Financial support from the Alexander von Humboldt Foundation, under the programme financed by the German Federal Ministry of Education and Research entitled German Research Chair No 01DG15010. AT acknowledges support from the Deutsche Forschungsgemeinschaft (DFG) through the TUM International Graduate School of Science and Engineering (IGSSE), GSC 81, within the project GENOMIE-QADOP, and the TUM Global Incentive Fund (exchange grant with Ghana).

\newpage


\newpage
\appendix
\section{Proof of Theorem 3.4}\label{GAS_DFE}

\begin{proof}
	The system (\ref{deterministic}) can be written as:
	\begin{equation}
		\begin{array}{rl}
			\dfrac{\mathrm{d}a}{\mathrm{d}t}&=(F-V)a-f(a,b),\\\\
			\dfrac{\mathrm{d}b}{\mathrm{d}t}&=g(a,b),
		\end{array}
	\end{equation} 
	where $a=(I_1,I_2)^{T}$ is the vector representing the infected classes, $b=(S_1,S_2,R)^{T}$ is the vector representing the uninfected classes, the matrices $F$ and $V$ are given as in Equation (\ref{FV}) and
	$$f(a,b)= \begin{bmatrix} \beta_{11}\Bigg(\dfrac{S_1^{0}}{N^{0}}-\dfrac{S_1}{N}\Bigg)I_1+\beta_{12}\Bigg(\dfrac{S_1^{0}}{N^{0}}-\dfrac{S_1}{N}\Bigg)I_2\\\\ \beta_{21}\Bigg(\dfrac{S_2^{0}}{N^{0}}-\dfrac{S_2}{N}\Bigg)I_1+\beta_{22}\Bigg(\dfrac{S_2^{0}}{N^{0}}-\dfrac{S_2}{N}\Bigg)I_2\end{bmatrix} \hspace{.2cm}\text{and}\hspace{.2cm}g(a,b)= \begin{bmatrix} \theta(1-p)+\lambda_1S_1-\mu S_1\\\\ \theta p+\lambda_2S_2-\mu S_2\\\\ \gamma_1I_1+\gamma_2I_2-\mu R\end{bmatrix}.$$ 
	Then, 
	$$V^{-1}F=\begin{bmatrix} \dfrac{\beta_{11}S_1^{0}}{N^{0}(\mu + \gamma_1 +d_1)}& \dfrac{\beta_{12}S_1^{0}}{N^{0}(\mu + \gamma_1 +d_1)}\\\\\dfrac{\beta_{21}S_2^{0}}{N^{0}(\mu + \gamma_2 +d_2)} & \dfrac{\beta_{22}S_2^{0}}{N^{0}(\mu + \gamma_2 +d_2)}\end{bmatrix}, $$
	and the left eigenvector of $V^{-1}F$, $(\omega_1,\omega_2)$ associated with the eigenvalue $\mathcal{R}_0$ is given by:
	
	$\omega_1=1$ and $\omega_2=\dfrac{N^{0}(\mu + \gamma_2 +d_2)}{\beta_{21}S_2^{0}}\Bigg(\mathcal{R}_0- \dfrac{\beta_{11}S_1^{0}}{N^{0}(\mu + \gamma_1 +d_1)}\Bigg)$ since
	
	$(\omega_1,\omega_2)V^{-1}F=\mathcal{R}_0(\omega_1,\omega_2).$
	
	Let us consider the following Lyapunov function:
	\begin{align}
		Q =& (\omega_1,\omega_2)V^{-1}(I_1,I_2)^{T} \nonumber\\
		=& \dfrac{I_1}{\mu + \gamma_1 +d_1}+\Bigg(\mathcal{R}_0- \dfrac{\beta_{11}S_1^{0}}{N^{0}(\mu + \gamma_1 +d_1)}\Bigg)\dfrac{N^{0}I_2}{\beta_{21}S_2^{0}}.
	\end{align}
	Then the derivative of $Q$ with respect to $t$ yields, $$Q^{'}=(\mathcal{R}_0-1)(\omega_1,\omega_2)^{T}a-(\omega_1,\omega_2)^{T}V^{-1}f(a,b).$$
	Since $(\omega_1,\omega_2)\geqslant 0,$ $V^{-1}\geqslant 0$ and $f(a,b)\geqslant0$ in $\Omega$, then $(\omega_1,\omega_2)^{T}V^{-1}f(a,b)\geqslant0.$
	Therefore, $Q^{'}\leqslant 0$ in $\Omega$ if $\mathcal{R}_0\leqslant 1$ and $Q$ is a Lyapunov function for the system (\ref{deterministic}). By LaSalle’s invariance principle \cite{lasallestability,lasalle1976stability}, $Q^{0}$ is GAS in $\Omega$.
	
	If $\mathcal{R}_0>1,$ then $Q^{'}=(\mathcal{R}_0-1)(\omega_1,\omega_2)^{T}a>0$ provided that $a>0$ and $b=(S_1^{0},S_2^{0},0)$. By continuity, $Q^{'}>0$ in the neighborhood of $Q^{0}$. Solutions in positive cone sufficiently close to $Q^{0}$ move away from $Q^{0}$, implying that $Q^{0}$ is unstable. Thus, the model system (\ref{deterministic}) is uniformly persistent \cite{freedman1994uniform, li1999global}. Uniform persistence and the positively invariance of $\Omega$ imply the existence of an endemic equilibrium.
\end{proof}

\section{Proof of Theorem 3.5}\label{GAS_EE}
{
\begin{proof}
	Consider the following Lyapunov candidate function:
	
	$$L=L_{1}+L_{2}+L_{3}+L_{4},$$where $L_1=S_{1}-S_1^{*}-S_1^{*}\log\Bigg(\dfrac{S_1}{S_1^{*}}\Bigg)$, $L_2=S_{2}-S_2^{*}-S_2^{*}\log\Bigg(\dfrac{S_2}{S_2^{*}}\Bigg)$, $L_3=I_{1}-I_1^{*}-I_1^{*}\log\Bigg(\dfrac{S_3}{S_3^{*}}\Bigg)$ and $L_4=I_{2}-I_2^{*}-I_2^{*}\log\Bigg(\dfrac{I_4}{I_4^{*}}\Bigg)$.
	
	Using the inequality $1-z+\log(z)\leqslant 0$ for $z>0$ with equality if and only if $z = 1$,
	differentiation and using the EE values give
	$$L^{'}=L_{1}^{'}+L_{2}^{'}+L_{3}^{'}+L_{4}^{'},$$
	where
	
	\begin{equation}\label{D1}
		\begin{array}{llll}
			L_1^{'}&=\Bigg(1-\dfrac{S^{*}_{1}}{S_1}\Bigg)\dfrac{\mathrm{d}S_{1}}{\mathrm{d}t}\\\\
			&=\Bigg(1-\dfrac{S^{*}_{1}}{S_1}\Bigg)\Bigg[\beta_{11}\dfrac{S_1^{*}I_1^{*}}{N^{*}}-\beta_{11}\dfrac{S_1I_1}{N}+\beta_{12}\dfrac{S_1^{*}I_2^{*}}{N^{*}}-\beta_{12}\dfrac{S_1I_2}{N}-\mu S_1+\mu S^{*}_{1}\Bigg]\\\\
			&=-\dfrac{\mu(S_1-S^{*}_{1})^{2}}{S_1}+\beta_{11}\dfrac{S_1^{*}I_1^{*}}{N^{*}}\Bigg[ 1-\dfrac{S^{*}_{1}}{S_1}-\dfrac{S_1I_1N^{*}}{S_1^{*}I_1^{*}N}+\dfrac{I_1N^{*}}{I_1^{*}N}\Bigg]+\beta_{12}\dfrac{S_1^{*}I_2^{*}}{N^{*}}\Bigg[ 1-\dfrac{S^{*}_{1}}{S_1}-\dfrac{S_1I_2N^{*}}{S_1^{*}I_2^{*}N}+\dfrac{I_2N^{*}}{I_2^{*}N}\Bigg].\\\\
		\text{Then} \quad	L_1^{'}	\leqslant&\beta_{11}\dfrac{S_1^{*}I_1^{*}}{N^{*}}\Bigg[\dfrac{I_1N^{*}}{I_1^{*}N}-\log\Bigg(\dfrac{I_1N^{*}}{I_1^{*}N}\Bigg)- \dfrac{S_1I_1N^{*}}{S_1^{*}I_1^{*}N}+ \log\Bigg(\dfrac{S_1I_1N^{*}}{S_1^{*}I_1^{*}N}\Bigg)
			\Bigg]\\\\
			&+\beta_{12}\dfrac{S_1^{*}I_2^{*}}{N^{*}}\Bigg[\dfrac{I_2N^{*}}{I_2^{*}N}-\log\Bigg(\dfrac{I_2N^{*}}{I_2^{*}N}\Bigg)
			-\dfrac{S_1I_2N^{*}}{S_1^{*}I_2^{*}N}+ \log\Bigg(\dfrac{S_1I_2N^{*}}{S_1^{*}I_2^{*}N}\Bigg)\Bigg].
		\end{array}
	\end{equation}	
We can also deduce in an analogous way:
	\begin{equation}\label{D2}
		\begin{array}{llll}
			L_2^{'}	&\leqslant\beta_{22}\dfrac{S_2^{*}I_2^{*}}{N^{*}}\Bigg[\dfrac{I_2N^{*}}{I_2^{*}N}-
			\log\Bigg(\dfrac{I_2N^{*}}{I_2^{*}N}\Bigg)- \dfrac{S_2I_2N^{*}}{S_2^{*}I_2^{*}N}+\log\Bigg(\dfrac{S_2I_2N^{*}}{S_2^{*}I_2^{*}N}\Bigg)\Bigg]\\\\
			&+\beta_{21}\dfrac{S_2^{*}I_1^{*}}{N^{*}}\Bigg[\dfrac{I_1N^{*}}{I_1^{*}N}-\log\Bigg(\dfrac{I_1N^{*}}{I_1^{*}N}\Bigg)- \dfrac{S_2I_1N^{*}}{S_2^{*}I_1^{*}N}+\log\Bigg(\dfrac{S_2I_1N^{*}}{S_2^{*}I_1^{*}N}\Bigg)\Bigg]. \end{array}
	\end{equation}
	
We also have
	
	\begin{equation}\label{D3}
		\begin{array}{llll}	L_3^{'}&=\Bigg(1-\dfrac{I^{*}_{1}}{I_1}\Bigg)\dfrac{\mathrm{d}I_{1}}{\mathrm{d}t}\\\\
			&=\Bigg(1-\dfrac{I^{*}_{1}}{I_1}\Bigg)\Bigg[\beta_{11}\dfrac{S_1I_1}{N}+\beta_{12}\dfrac{S_1I_2}{N}-(\mu +\gamma_1+d_1)I_1 \Bigg]\\\\
			&=\Bigg(1-\dfrac{I^{*}_{1}}{I_1}\Bigg)\Bigg[\beta_{11}\dfrac{S_1I_1}{N}+\beta_{12}\dfrac{S_1I_2}{N}-\beta_{11}\dfrac{S_1^{*}I_1}{N^{*}}+\beta_{12}\dfrac{S_1^{*}I_2^{*}I_1}{N^{*}I_1^{*}}\Bigg]\\\\
			&=\beta_{11}\dfrac{S_1^{*}I_1^{*}}{N^{*}}\Bigg[\dfrac{S_1I_1N^{*}}{S_1^{*}I_1^{*}N}-\dfrac{S_1N^{*}}{S_1^{*}N}-\dfrac{I_1}{I_1^{*}} +1\Bigg]+\beta_{12}\dfrac{S_1^{*}I_2^{*}}{N^{*}}\Bigg[\dfrac{S_1I_2N^{*}}{S_1^{*}I_2^{*}N}-\dfrac{S_1I_1^{*}I_2N^{*}}{S_1^{*}I_1I_2^{*}N}-\dfrac{I_1}{I_1^{*}} +1\Bigg],\\\\
			L_3^{'}&\leqslant\beta_{11}\dfrac{S_1^{*}I_1^{*}}{N^{*}}\Bigg[\dfrac{S_1I_1N^{*}}{S_1^{*}I_1^{*}N}-\log\Bigg(\dfrac{S_1I_1N^{*}}{S_1^{*}I_1^{*}N}\Bigg)-\dfrac{I_1}{I_1^{*}}+\log\Bigg(\dfrac{I_1}{I_1^{*}}\Bigg)\Bigg]\\\\
			&+\beta_{12}\dfrac{S_1^{*}I_2^{*}}{N^{*}}\Bigg[\dfrac{S_1I_2N^{*}}{S_1^{*}I_2^{*}N}-\log\Bigg(\dfrac{S_1I_2N^{*}}{S_1^{*}I_2^{*}N}\Bigg)-\dfrac{I_1}{I_1^{*}}+\log\Bigg(\dfrac{I_1}{I_1^{*}}\Bigg)
			\Bigg]. \end{array}
	\end{equation}
	
Similarly, we obtain
	\begin{equation}\label{D4}
		\begin{array}{llll}
			L_4^{'}	&\leqslant\beta_{22}\dfrac{S_2^{*}I_2^{*}}{N^{*}}\Bigg[\dfrac{S_2I_2N^{*}}{S_2^{*}I_2^{*}N}-\log\Bigg(\dfrac{S_2I_2N^{*}}{S_2^{*}I_2^{*}N}\Bigg)-\dfrac{I_2}{I_2^{*}}+ln\dfrac{I_2}{I_2^{*}} \Bigg]\\\\
			&+\beta_{21}\dfrac{S_2^{*}I_1^{*}}{N^{*}}\Bigg[\dfrac{S_2I_1N^{*}}{S_2^{*}I_1^{*}N}-\log\Bigg(\dfrac{S_2I_1N^{*}}{S_2^{*}I_1^{*}N}\Bigg)-\dfrac{I_2}{I_2^{*}}+\log\Bigg(\dfrac{I_2}{I_2^{*}} \Bigg)\Bigg]. \end{array}
	\end{equation}
	Therefore, by adding (\ref{D1}), (\ref{D2}), (\ref{D3}) and (\ref{D4}) we deduce
	\begin{align*}
		L^{'}\leqslant& \Bigg(-\dfrac{I_1N^{*}}{I_1^{*}N}+\log\Bigg(\dfrac{I_1N^{*}}{I_1^{*}N} \Bigg)\Bigg)\Bigg(-\beta_{11}\dfrac{S_1^{*}I_1^{*}}{N^{*}}-\beta_{21}\dfrac{S_2^{*}I_1^{*}}{N^{*}}\Bigg)\\
		&+\Bigg(-\dfrac{I_2N^{*}}{I_2^{*}N}+\log\Bigg(\dfrac{I_2N^{*}}{I_2^{*}N} \Bigg)\Bigg)\Bigg(-\beta_{12}\dfrac{S_1^{*}I_2^{*}}{N^{*}}-\beta_{22}\dfrac{S_2^{*}I_2^{*}}{N^{*}}\Bigg)\\
		&+\Bigg(-\dfrac{I_1}{I_1^{*}}+\log\Bigg(\dfrac{I_1}{I_1^{*}} \Bigg)\Bigg)\Bigg(\beta_{11}\dfrac{S_1^{*}I_1^{*}}{N^{*}}+\beta_{12}\dfrac{S_1^{*}I_2^{*}}{N^{*}}\Bigg)	\\
		&+\Bigg(-\dfrac{I_2}{I_2^{*}}+\log\Bigg(\dfrac{I_2}{I_2^{*}} \Bigg)\Bigg)\Bigg(\beta_{22}\dfrac{S_2^{*}I_2^{*}}{N^{*}}+\beta_{21}\dfrac{S_2^{*}I_1^{*}}{N^{*}}\Bigg).\\
\text{Then} \quad L^{'}\leqslant& 0, \hspace{.2cm}\text{since}\hspace{.2cm} -z+\log(z)\leqslant-1, \hspace{.2cm} \forall z>0.
	\end{align*}	
	Since $\{Q^{*}\}$ is the only invariant subset in $\Omega$ where $L=0$, therefore by LaSalle’s invariance principle \cite{lasalle1976stability}, $Q^{*}$ is GAS in $\Omega$.
\end{proof}
}

\newpage
\section{Tables}\label{Tables}

\begin{table}
	\caption{PRCC of model’s parameters at time $t$ (days) with strong PI and weak efficiency of vaccine. The values $\theta=1000$, $\mu= 0.09$, $\beta_{11}=0.35$, $\beta_{12}=0.28$, $\beta_{21}=0.175$, $\beta_{22}=0.14$ are used as baseline. } \label{sensibilité1}\begin{tabular}{c} 
		\begin{tabular}{c|ccc|rrr} \hline
			Parameters & \multicolumn{3}{c|}{ Range of parameters}& \multicolumn{3}{c}{ Total Infected: $I_1+I_2$}\\
			\cline{2-7} &  Min & Baseline& Max & $t=50$ days&$t=100$ days& $t=200$ days\\
			\hline \hline 	
			$\theta$&$500$&$1000$&$ 1500 $&$0.71395^{**}$&$ 0.78511^{**}$&$0.76166^{***}$\\
			$ p$  & $0$ &$ 0.5$&$1$ &$0.020314$&$0.0029584$&$0.028397$ \\
			$\beta_{11}$ &$0.175$ & $0.35$&$  0.525$&$ 0.85757^{***}$&$0.8731^{***}$&$0.87175^{***}$\\
			$\beta_{12}$&$0.14$ & $0.28$&$ 0.42$&$ 0.0047432 $&$0.027724$&$-0.030496$\\
			$\beta_{21}$ & $0.0875$ & $0.175$&$0.2625 $&$ 0.0090246$&$-0.012341$&$0.026579$\\
			$\beta_{22}$ & $0.07$ & $0.14$ & $0.21$ & $-0.047262$ & $0.02905$&$-0.037461$\\
			$\mu$  &  $0.045$ & $0.09$&$0.135$&$ -0.7695^{**}$&$ -0.80652^{***} $&$-0.79222^{**}$\\
			$d_{1}$ &$0.0004$& $0.0008$&$ 0.0012$&  $-0.012188$&$0.03368$&$-0.046922$\\
			$d_{2}$ & $0.00005$& $0.0001$&$0.00015$&$-0.025215$&$0.016188$&$-0.043869$\\
			$\gamma_{1}$ & $0.05$& $0.1$&$0.15$&$-0.78315^{**}$&$ -0.84015^{***}$&$-0.82903^{***}$\\
			$\gamma_{2}$ & $0.0625$& $0.13$&$0.1925$&$ 0.010702$&$0.05007$&$0.012449$\\
			\hline
		\end{tabular} 
	\end{tabular}\\
	**: PRCC values: $~0.7$ to $ 0.79
	$ or $-0.7$ to $-0.79$; ***: PRCC values: $~0.8$ to $ 0.99$ or $-0.8$ to $-0.99$
\end{table}

\begin{table}
	\caption{PRCC of model’s parameters at time $t$ days with strong PI and strong efficiency of vaccine, when $\theta=1000$, $\mu= 0.09$, $\beta_{11}=0.35$, $\beta_{12}=0.28$, $\beta_{21}=0.035$, $\beta_{22}=0.028$ as baseline. } \label{sensibilité2}\begin{tabular}{c} 
		\begin{tabular}{c|ccc|rrr} \hline
			Parameters & \multicolumn{3}{c|}{ Range of parameters}& \multicolumn{3}{c}{ Total Infected: $I_1+I_2$}\\
			\cline{2-7} &  Min & Baseline& Max & $t=50$ days&$t=100$ days& $t=200$ days\\
			\hline \hline 	
			$\theta$&$500$&$1000$&$ 1500 $&$0.7381^{**}$&$ 0.76387^{**}$&$0.78486^{**}$\\
			$ p$  & $0$ &$ 0.5$&$1$ &$0.003905$&$0.0037356$&$0.027257$ \\
			$\beta_{11}$ &$0.175$ & $0.35$&$  0.525$&$ 0.86469^{***}$&$0.87427^{***} $&$0.88181^{***}$\\
			$\beta_{12}$&$0.14$ & $0.28$&$ 0.42$&$0.0079816 $&$0.033516$&$0.030158$\\
			$\beta_{21}$ & $0.0175$ & $0.035$&$0.0525 $&$0.0012134$&$0.018087$&$-0.00058438$\\
			$\beta_{22}$ & $0.014$ & $0.028$ & $0.042$ & $0.021841$ & $-0.0038364$&$0.018881$\\
			$\mu$  &  $0.045$ & $0.09$&$0.135$&$ -0.78849^{**}$&$ -0.80304^{***} $&$-0.81535^{***}$\\
			$d_{1}$ &$0.0004$& $0.0008$&$ 0.0012$&  $-0.054627$&$ 0.066816$&$0.019678$\\
			$d_{2}$ & $0.00005$& $0.0001$&$0.00015$&$-0.033227$&$-0.021472$&$-0.028882$\\
			$\gamma_{1}$ & $0.05$& $0.1$&$0.15$&$-0.80324^{***}$&$  -0.83346^{***}$&$-0.84421^{***}$\\
			$\gamma_{2}$ & $0.0625$& $0.13$&$0.1925$&$ -0.0099732$&$-0.02272$&$0.0020891$\\
			\hline
		\end{tabular} 
	\end{tabular}\\
	**: PRCC values: $~0.7$ to $ 0.79$ or $-0.7$ to $-0.79$; ***: PRCC values: $~0.8$ to $ 0.99
	$ or $-0.8$ to $-0.99$
\end{table}

\begin{table}
	\caption{PRCC of model’s parameters at time $t$ days with weak PI and weak efficiency of vaccine, when $\theta=10$, $\mu= 0.0009$, $\beta_{11}=0.35$, $\beta_{12}=0.28$, $\beta_{21}=0.175$, $\beta_{22}=0.14$ as baseline. } \label{sensibilité3}\begin{tabular}{c} 
		\begin{tabular}{c|ccc|rrr} \hline
			Parameters & \multicolumn{3}{c|}{ Range of parameters}& \multicolumn{3}{c}{ Total Infected: $I_1+I_2$}\\
			\cline{2-7} &  Min & Baseline& Max & $t=50$ days&$t=100$ days& $t=200$ days\\
			\hline \hline 	
			$\theta$&$5$&$10$&$ 15 $&$0.45737$&$ 0.37556$&$0.51163^{*}$\\
			$ p$  & $0$ &$ 0.5$&$1$ &$-0.041574$&$0.028378$&$0.030938$ \\
			$\beta_{11}$ &$0.175$ & $0.35$&$  0.525$&$ -0.63334^{*}$&$-0.23892$&$0.355$\\
			$\beta_{12}$&$0.14$ & $0.28$&$ 0.42$&$ -0.23979 $&$-0.24053$&$-0.13989$\\
			$\beta_{21}$ & $0.0875$ & $0.175$&$0.2625 $&$ -0.90072^{***} $&$-0.90502^{***}$&$-0.80837^{***}$\\
			$\beta_{22}$ & $0.07$ & $0.14$ & $0.21$ & $-0.52059^{*}$ & $-0.50519^{*}$&$-0.30843$\\
			$\mu$  &  $0.00045$ & $0.0009$&$0.00135$&$ -0.031697$&$ -0.18722 $&$-0.15951$\\
			$d_{1}$ &$0.0004$& $0.0008$&$ 0.0012$&  $0.012078$&$-0.038623$&$0.01511$\\
			$d_{2}$ & $0.00005$& $0.0001$&$0.00015$&$0.028409$&$0.0088495$&$0.047733$\\
			$\gamma_{1}$ & $0.05$& $0.1$&$0.15$&$-0.12428$&$  0.81303^{***}$&$0.59284^{*}$\\
			$\gamma_{2}$ & $0.0625$& $0.13$&$0.1925$&$ 0.48726$&$0.62754^{*}$&$0.48082$\\
			\hline
		\end{tabular} 
	\end{tabular}\\
	*: PRCC values: $~0.5$ to $ 0.69
	$ or $-0.5$ to $-0.69$; ***: PRCC values: $~0.8$ to $ 0.99
	$ or $-0.8$ to $-0.99$
\end{table}

\begin{table}
	\caption{PRCC of model’s parameters at time $t$ days with weak PI and strong efficiency of vaccine, when $\theta=10$, $\mu= 0.0009$, $\beta_{11}=0.35$, $\beta_{12}=0.28$, $\beta_{21}=0.035$, $\beta_{22}=0.028$ as baseline. } \label{sensibilité4}\begin{tabular}{c} 
		\begin{tabular}{c|ccc|rrr} \hline
			Parameters & \multicolumn{3}{c|}{ Range of parameters}& \multicolumn{3}{c}{ Total Infected: $I_1+I_2$}\\
			\cline{2-7} &  Min & Baseline& Max & $t=50$ days&$t=100$ days& $t=200$ days\\
			\hline \hline 	
			$\theta$&$5$&$10$&$ 15 $&$0.5751^{*}$&$ 0.48818$&$0.61256^{*}$\\
			$ p$  & $0$ &$ 0.5$&$1$ &$0.052835$&$0.014154$&$0.050557$ \\
			$\beta_{11}$ &$0.175$ & $0.35$&$  0.525$&$ -0.70943^{**}$&$-0.47854 $&$0.44458$\\
			$\beta_{12}$&$0.14$ & $0.28$&$ 0.42$ &$-0.16357$&$-0.16371$&$-0.03405$\\
			$\beta_{21}$ & $0.0175$ & $0.035$&$0.0525 $&$ -0.84854^{***}$&$-0.90973^{***}$&$-0.85731^{***}$\\
			$\beta_{22}$ & $0.014$ & $0.028$ & $0.042$ & $-0.17909$ & $-0.22613$&$-0.14446$\\
			$\mu$  &  $0.00045$ & $0.0009$&$0.00135$&$ -0.40329$&$ -0.61646^{*} $&$-0.72754^{**}$\\
			$d_{1}$ &$0.0004$& $0.0008$&$ 0.0012$&  $-0.072168$&$ 0.04039$&$-0.040258$\\
			$d_{2}$ & $0.00005$& $0.0001$&$0.00015$&$0.019586$&$-0.053637$&$-0.030518$\\
			$\gamma_{1}$ & $0.05$& $0.1$&$0.15$&$-0.81298^{***}$&$  0.76378^{**}$&$0.69479^{*}$\\
			$\gamma_{2}$ & $0.0625$& $0.13$&$0.1925$&$ 0.20028$&$0.31528$&$0.2891$\\
			\hline
		\end{tabular} 
	\end{tabular}\\
	*: PRCC values: $~0.5$ to $ 0.69
	$ or $-0.5$ to $-0.69$; **: PRCC values: $~0.7$ to $ 0.79
	$ or $-0.7$ to $-0.79$; \hspace{2cm}***: PRCC values: $~0.8$ to $ 0.99
	$ or $-0.8$ to $-0.99$
\end{table}

\newpage
\section{Figures}

\begin{figure}[!h]
	\centering 
	\includegraphics[width=.7\textwidth]{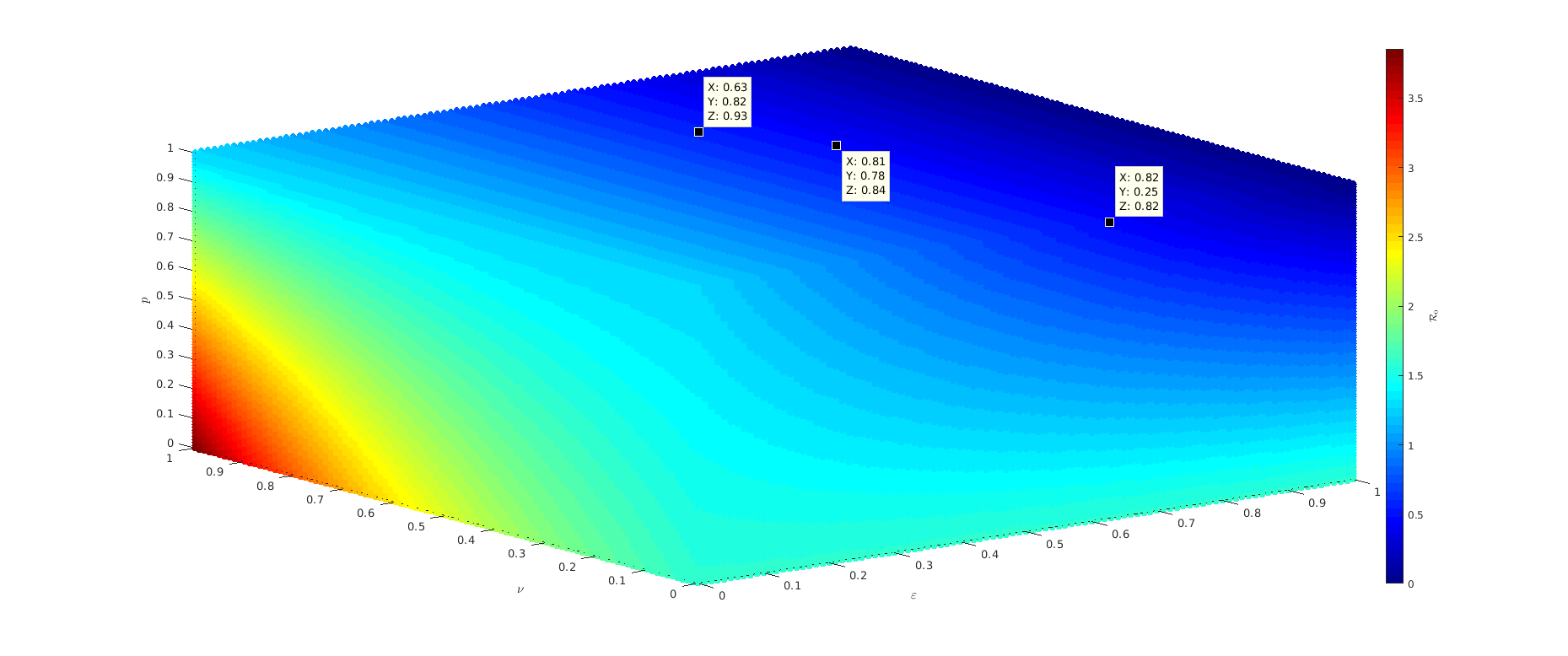}
	\caption{Scatter plots of $\mathcal{R}_0$ with a strong turnover as a function of $\varepsilon$, $\nu$ and $p$ .The parameters are $\theta=1000$, $\beta_{11}= 0.35$, $\beta_{12}=0.28$, $\beta_{21}=0.175, \beta_{22}=0.14$, $\mu=0.09$, $d_1=0.0008$, $d_2=0.0001$, $\gamma_{1}=0.065$, $\gamma_{2}=0.13$.}\label{ScatterST}
\end{figure}

\begin{figure}[!h]
	\centering 
	\includegraphics[width=.7\textwidth]{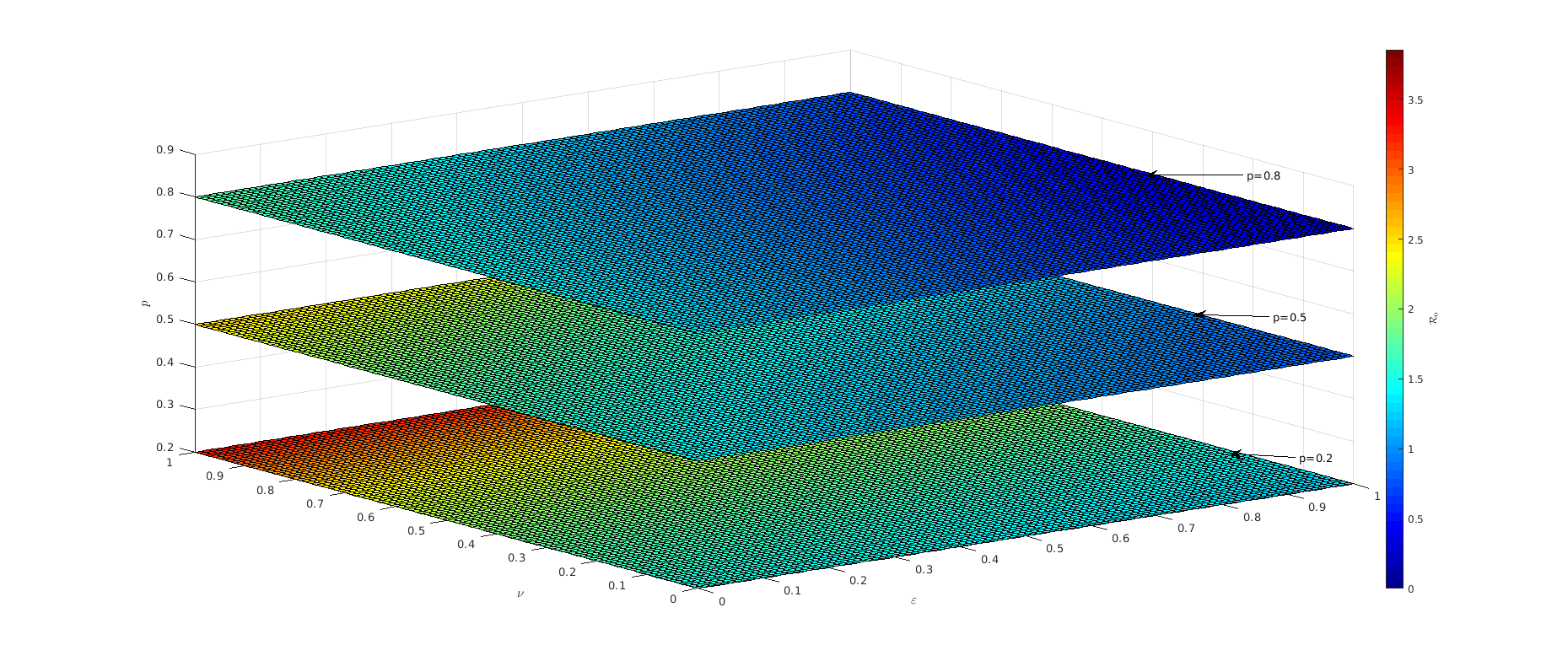}
	\caption{Slice planes of $\mathcal{R}_0$ orthogonal to the p-axis at the values 0.2, 0.5, 0.8 with a strong turnover. The parameters are $\theta=1000$, $\beta_{11}= 0.35$, $\beta_{12}=0.28$, $\beta_{21}=0.175, \beta_{22}=0.14$, $\mu=0.09$, $d_1=0.0008$, $d_2=0.0001$, $\gamma_{1}=0.065$, $\gamma_{2}=0.13$.}\label{SliceST}
\end{figure}

\begin{figure}[!h]
	\centering
	\begin{subfigure}[b]{0.48\textwidth}
		\includegraphics[width=1.1\linewidth]{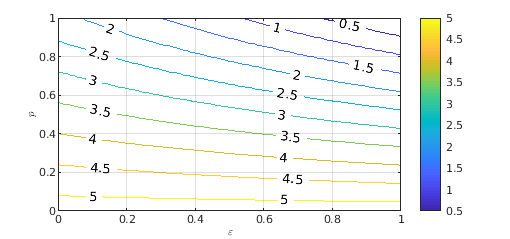}
		\caption{}\label{ContourWT1}
	\end{subfigure}
	\begin{subfigure}[b]{0.48\textwidth}
		\includegraphics[width=1.1\linewidth]{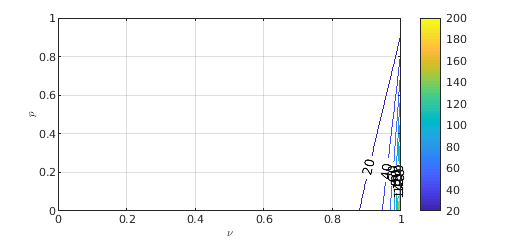}
		\caption{}\label{ContourWT2}
	\end{subfigure}
	
	\begin{subfigure}[b]{0.48\textwidth}
		\includegraphics[width=1.1\linewidth]{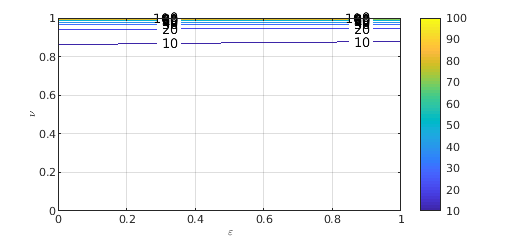}
		\caption{}\label{ContourWT3}
	\end{subfigure}	
	\caption{Contour plots of the basic reproduction number ($\mathcal{R}_0$) of the model (\ref{deterministic}) with a weak turnover as a function of: (a) vaccine coverage, $p$, and vaccine efficiency on the transmission, $\varepsilon$ (fixed $\nu=0.5$); (b) vaccine coverage, $p$, and vaccine efficiency on the ability to enhance recovery, $\nu$ (fixed $\varepsilon=0.5$); (c) vaccine efficiency on the ability of being recovered, $\nu$, and vaccine efficiency on the transmission, $\varepsilon$ (fixed $p=0.5$). The parameters are $\theta=10$, $\beta_{11}= 0.35$, $\beta_{12}=0.28$, $\beta_{21}=0.175, \beta_{22}=0.14$, $\mu=0.0009$, $d_1=0.0008$, $d_2=0.0001$, $\gamma_{1}=0.065$, $\gamma_{2}=0.13$}\label{ContourWT}
\end{figure}

\begin{figure}[!h]
	\centering 
	\includegraphics[width=.7\textwidth]{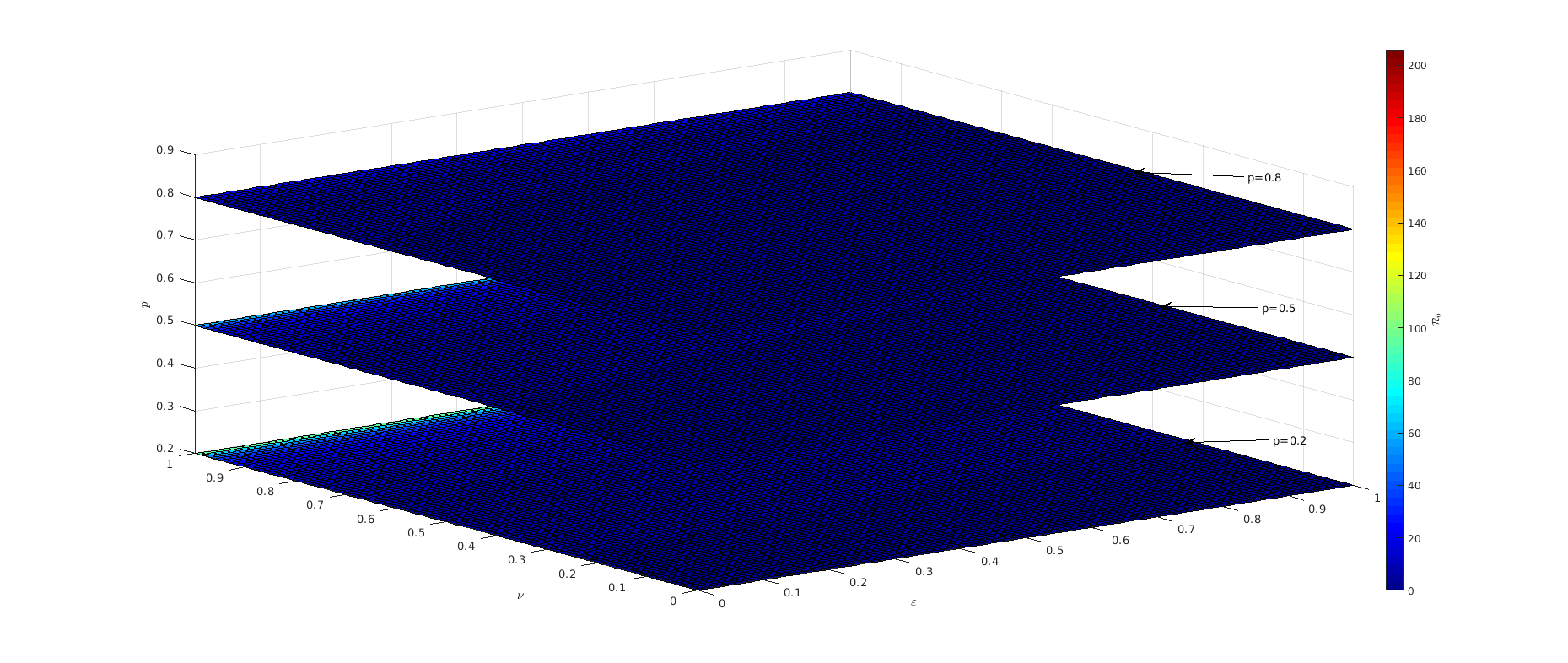}
	\caption{Slice planes of $\mathcal{R}_0$ orthogonal to the p-axis at the values 0.2, 0.5, 0.8 with a weak turnover. The parameters are $\theta=10$, $\beta_{11}= 0.35$, $\beta_{12}=0.28$, $\beta_{21}=0.175, \beta_{22}=0.14$, $\mu=0.0009$, $d_1=0.0008$, $d_2=0.0001$, $\gamma_{1}=0.065$, $\gamma_{2}=0.13$.}\label{SliceWT}
\end{figure}

\begin{figure}[!h]
	\centering 
	\includegraphics[width=.7\textwidth]{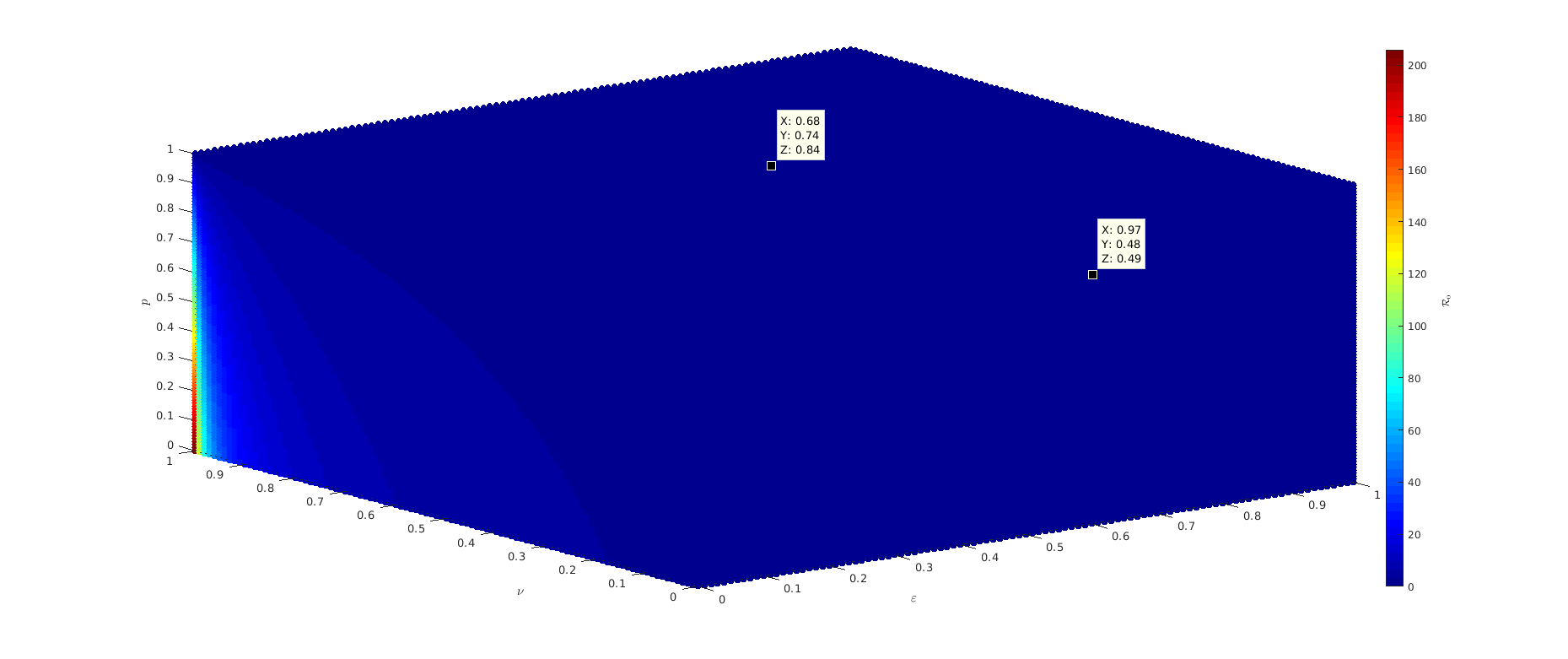}
	\caption{Scatter plots of $\mathcal{R}_0$ with a weak turnover as a function of $\varepsilon$, $\nu$ and $p$ .The parameters are $\theta=10$, $\beta_{11}= 0.35$, $\beta_{12}=0.28$, $\beta_{21}=0.175, \beta_{22}=0.14$, $\mu=0.0009$, $d_1=0.0008$, $d_2=0.0001$, $\gamma_{1}=0.065$, $\gamma_{2}=0.13$.}\label{ScatterWT}
\end{figure}


\begin{figure}[!h]
	\centering
	\begin{subfigure}[b]{0.48\textwidth}
		\includegraphics[width=1.1\linewidth]{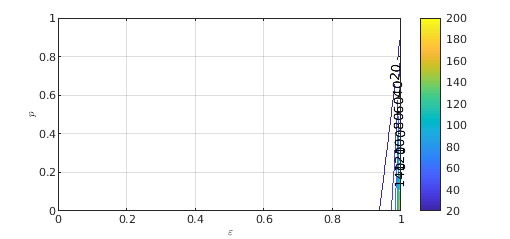}
		\caption{}\label{contourconveWT}
	\end{subfigure}
	\begin{subfigure}[b]{0.48\textwidth}
		\includegraphics[width=1.1\linewidth]{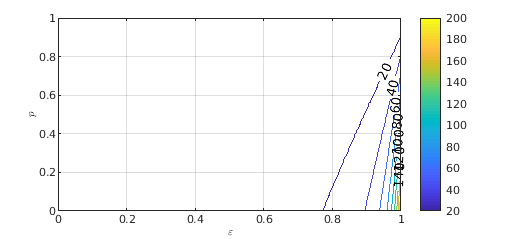}
		\caption{}\label{contourconcaWT}
	\end{subfigure}
	
	\begin{subfigure}[b]{0.48\textwidth}
		\includegraphics[width=1.1\linewidth]{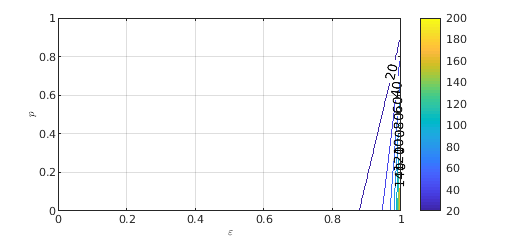}
		\caption{}\label{contourlineaWT}
	\end{subfigure}	
	\caption{Contour plot of the basic reproduction number ($\mathcal{R}_0$) of the model (\ref{deterministic}) with a weak turnover as a function of vaccine coverage, $p$, and vaccine efficiency on the transmission, $\varepsilon$ when: (a)  $\nu=\varepsilon^{2}$(convex relationship); (b) $\nu=\sqrt{\varepsilon}$(concave relationship);   (c)$\nu=\varepsilon$(linear relationship). The parameters are $\theta=1000$, $\beta_{11}= 0.35$, $\beta_{12}=0.28$, $\beta_{21}=0.175, \beta_{22}=0.14$, $\mu=0.09$, $p=0.5$, $d_1=0.0008$, $d_2=0.0001$, $\gamma_{1}=0.065$, $\gamma_{2}=0.13$.}
\end{figure}

\end{document}